\documentclass[DIV=classic,fontsize=11pt]{scrartcl}

\usepackage[latin1]{inputenc}
\usepackage{amsmath,amssymb,amsthm}
\usepackage{enumerate}
\usepackage[numbers,sort&compress]{natbib}
\usepackage[left=2.7cm,top=2.5cm,right=2.8cm,bottom= 5cm]{geometry}

\numberwithin{equation}{section}

\newtheorem{theorem}{Theorem}[section]

\newtheorem{proposition}[theorem]{Proposition}
\newtheorem{corollary}[theorem]{Corollary}
\newtheorem{lemma}[theorem]{Lemma}
\newtheorem{remark}[theorem]{Remark}
\newtheorem{example}[theorem]{Example}
\newtheorem{examples}[theorem]{Examples}
\newtheorem{foo}[theorem]{Remarks}

\newcommand{\1}{\mathbf 1}

\newcommand{\N}{\mathbb{N}}

\newcommand{\R}{\mathbb{R}}

\renewcommand{\d}{\mathrm{d}}
\renewcommand{\epsilon}{\varepsilon}
\newcommand{\dd}{\,\mathrm{d}}
\newcommand{\var}{\text{-}\mathrm{var}}

\begin{document}

\title{Duality for pathwise superhedging in continuous time}

\author{Daniel Bartl\thanks{D.B. has been funded by the Vienna Science and Technology Fund (WWTF) through project VRG17-005 and by the Austrian Science Fund (FWF) under grant Y00782.},
Michael Kupper,
David J. Pr\"omel\thanks{D.J.P. gratefully acknowledges financial support of the Swiss National Foundation under Grant No.~$200021\_163014$ and was employed at ETH Z\"urich when this project was commenced.},
and Ludovic Tangpi\thanks{L.T. gratefully acknowledges financial support of the Vienna Science and Technologie Fund (WWTF) under project MA14-008.}}

\date{\today}

\maketitle

\begin{abstract}
  \noindent
  \textbf{Abstract.} We provide a model-free pricing-hedging duality in continuous time. For a frictionless market consisting of $d$~risky assets with continuous price trajectories, we show that the purely analytic problem of finding the minimal superhedging price of a path dependent European option has the same value as the purely probabilistic problem of finding the supremum of the expectations of the option over all martingale measures. The superhedging problem is formulated with simple trading strategies, the claim is the limit inferior of continuous functions, which allows for upper and lower semi-continuous claims, and superhedging is required in the pathwise sense on a $\sigma$-compact sample space of price trajectories. If the sample space is stable under stopping, the probabilistic problem reduces to finding the supremum over all martingale measures with compact support. As an application of the general results we deduce dualities for Vovk's outer measure and semi-static superhedging with finitely many securities.\\[2mm]  
  %
\end{abstract}

\section{Introduction}

Given the space $C([0,T],\mathbb{R}^d)$ of all continuous price trajectories, the superhedging problem of a contingent claim $X\colon C([0,T],\mathbb{R}^d)\to\mathbb{R}$ consists of finding the infimum over all $\lambda\in\mathbb{R}$ such that there exists a trading strategy~$H$ which satisfies
\begin{equation}\label{eq:hedging0}
  \lambda+ (H\cdot S)_T(\omega) \ge X(\omega),\quad \omega \in C([0,T],\mathbb{R}^d),
\end{equation}
where $(H\cdot S)_T(\omega)$ denotes the capital gain by trading according to the strategy~$H$ in the underlying assets $S_t(\omega):=\omega(t)$. 

In the classical framework of mathematical finance one commonly postulates a model for the price evolution by fixing a probability measure~$P$ such that~$S$ is a semimartingale and defines $(H\cdot S)_T$ as the stochastic integral~$\int_0^T H_t\dd S_t$. Then, a consequence of the fundamental theorem of asset pricing states that the infimum over all $\lambda$ such that there are admissible predictable integrands~$H$ fulfilling inequality~\eqref{eq:hedging0} is equal to the supremum of $E_Q[X]$ over all absolutely continuous local martingale measures $Q$, see \citet{Delbaen2006}. Here, the superhedging (i.e. inequality~\eqref{eq:hedging0}) is assumed to hold $P$-almost surely and the set of absolutely continuous local martingale measures is non-empty, which is guaranteed by the exclusion of some form of arbitrage, see \cite{Delbaen2006} for the precise formulation. 

More recently, alternative possibilities to specify the superhedging requirement without referring to a fixed model have been proposed. For instance, if an investor takes into account a class~$\mathcal{P}$ of probabilistic models, then superhedging is naturally required to hold $\mathcal{P}$-quasi surely, i.e.~$P$-almost surely for all considered models $P\in\mathcal{P}$. The pioneering works of \citet{Lyons1995} and \citet{Avellaneda1995} on Knightian uncertainty in mathematical finance consider models with uncertain volatility in continuous time. The study of the pricing-hedging duality in this setting has given rise to a rich literature starting with the capacity-theoretic approach of \citet{denis06}. Further, \citet{Peng2010} obtains the duality using stochastic control techniques, whereas \citet{STZ1,STZ3,STZ10} rely on supermartingale decomposition results under individual models and eventually build on aggregation results to derive the duality under model uncertainty. This approach has been extended by \citet{Neufeld2013} to cover measurable claims using the theory of analytic sets, see also \citet{Biagini2017} for a robust fundamental theorem of asset pricing under a model ambiguity version of the no-arbitrage of the first kind condition $\mathrm{NA}_1(\mathcal{P})$, and \citet{Nutz2015} for the case of jump diffusions.

In the present work we focus on the pathwise/model-free approach and assume that the superhedging requirement~\eqref{eq:hedging0} has to hold pointwise for all price trajectories in a given set $\Omega\subseteq C([0,T],\mathbb{R}^d)$. In this pathwise setting, finding the minimal superhedging price turns out to be a purely analytic problem and its formulation is independent of the probabilistic problem of finding the supremum of the expectation over (a subset of) all martingale measures. This is in contrast to the above mentioned approaches working with a fixed model, under Knightian uncertainty or in a quasi-sure setting. Notice that the pathwise approach corresponds to the quasi-sure approach when $\mathcal{P}$ contains all Dirac measures, which in continuous time is excluded, see e.g.~\cite[Corollary~3.5]{Biagini2017}. 

In the now classical paper~\cite{Hob98}, Hobson first addressed the problem of pathwise superhedging for the lookback option. His analysis was based on some sharp pathwise martingale inequalities and has motivated \citet{bei-hl-pen} to introduce the martingale optimal transport problem in discrete time. Here, the investor takes static positions in some liquidly traded vanilla options and dynamic positions in the stocks. The rationale is that information on the price of options translates into the knowledge of some marginals of the martingale measures; see also \cite{Acciaio2016,BCKT,RobHeding,bur-fri-mag,burzoni2016pointwise,bei-nutz-tou} for further developments in this direction. In continuous time, the duality for the martingale optimal transport has been obtained by \citet{gal-hl-tou} and \citet{Possamai2013} in the quasi-sure setting. The pathwise formulation was studied by \citet{Dolinsky2014} using a discretization of the sample space. These results have been extended by \citet{Hou2015}, who, in particular, allow incorporation of investor's beliefs (of possible price paths) by relying on the notion of ``prediction set'' due to \citet{Mykland2003}. 

Following this consideration in our analysis, we also assume that the investor does not deem every continuous paths plausible but focuses on a prediction set $\Omega \subseteq C([0,T],\mathbb{R}^d)$ that is required to be $\sigma$-compact (i.e.~at most a countable union of compact sets) and define the pathwise superhedging problem on the sample space~$\Omega$. Moreover, restricting the set of possible price paths has the financially desirable effect of reducing the superhedging price. See also \citet{Aksamit2016} and \citet{Acciaio2017} for other treatments of belief and information in robust superhedging, and \citet{DS_Skorokhod} and \citet{Guo-Tan-Tou17} for extensions of the pathwise formulation to the Skorokhod space.

In the continuous time setting already the definition of a pathwise ``stochastic integral'' is a non-trivial issue. We circumvent this problem by working with simple strategies and consider as ``stochastic'' integral the pointwise limit inferior of pathwise integrals against simple strategies; an approach that was proposed by \citet{Perkowski2016} to define an outer measure allowing to study stochastic integration under model ambiguity. This outer measure is very similar in spirit to that of \citet{Vovk2012} and can be seen as the value of a pathwise superhedging problem, cf. Section~\ref{subsec:vovk} for details and \citet{Beiglbock2017} and \citet{Vovk2016} for existing duality results in this setting. 

Formally, we define the superhedging price of a contingent claim $X\colon\Omega\to[-\infty,+\infty]$ as the infimum over all $\lambda\in\mathbb{R}$ such that there exists a sequence~$(H^n)$ of simple strategies which satisfies
\begin{equation*}
  \lambda+ \liminf_{n\to\infty} (H^n\cdot S)_T(\omega) \ge X(\omega)\quad\mbox{for all }\omega\in\Omega
\end{equation*}
and the admissibility condition $\lambda+(H^n\cdot S)_t(\omega)\geq 0$ for all $n\in\mathbb{N}$, $\omega\in\Omega$, and $t\in[0,T]$.
If $X$ is the limit inferior of a sequence of continuous functions, then under the assumptions that $\Omega$ is $\sigma$-compact and contains all its stopped paths, we show that the superhedging price coincides with the supremum of $E_Q[X]$ over all martingale measures~$Q$. Furthermore, this duality is generalized to the case when $X$ is unbounded from above and when $\Omega$ does not contain all its stopped paths. In addition to providing a way around the technical difficulty posed by the definition of pathwise stochastic integrals, the superheding in terms of limit inferior turns out to be necessary to guarantee the duality on a sufficiently large space, see Remark~\ref{rem:counterexample} for a counterexample.

Our main contributions to the pathwise pricing-hedging duality in continuous time and finitely many risky assets are as follows: While in the current literature (see e.g. \cite{Hou2015,Dolinsky2014,Guo-Tan-Tou17}) pathwise duality results hold for uniformly continuous options, the proposed method allows for much less regular claims (including for example European options, Spread options, continuously and discretely monitoring Asian options, lookback options, certain types of barrier options, and options on the realized variance). In particular, this implies a duality for Vovk's outer measure on closed sets. A related duality result was given by~\citet{Vovk2016}, however, under an additional closure assumption on the set of the attainable outcomes. Moreover, our pricing-hedging duality holds for every prediction set~$\Omega$ which is $\sigma$-compact. Let us remark that the assumption of $\sigma$-compactness is an essential ingredient of the presented method to get the pricing-hedging duality. We will show in Section~\ref{subsec:examples} that typical price trajectories for various popular financial models such as local, stochastic or even rough volatility models belong to the $\sigma$-compact space of H\"older continuous functions. In the related work~\cite{Hou2015} the pricing-hedging duality holds for an approximate version of the superhedging price which requires the superhedging on an enlarged prediction set $\Omega^\varepsilon :=\{\omega\in C([0,T],\mathbb{R}^d)\,:\,\inf_{\omega^\prime\in\Omega} \|\omega-\omega^\prime\|_\infty\le\varepsilon\}\supset \Omega$ for any given $\varepsilon> 0$. 

\medskip
The article is organized as follows: In Section~\ref{sec:main results} we present the main results (Theorem~\ref{thm:main.integrals} and Theorem~\ref{thm:main.integrals.Z}) and some direct applications. Section~\ref{sec:discussion} contains a detailed discussion of feasible choices for the underlying sample space. The proofs of the main results are carried out in Section~\ref{sec:proof}. A criterion for the sample path regularity of stochastic processes and the construction of a counterexample are given in Appendix~\ref{sec:appendix}.

\section{Main results}\label{sec:main results}

Let $\Omega\subset C([0,T],\mathbb{R}^d)$ be a non-empty metric space where $T>0$ is a finite time horizon and $d\in\mathbb{N}$. The canonical process $S\colon[0,T]\times\Omega\to\mathbb{R}^d$ given by $S_t(\omega):=\omega(t)$ generates the raw filtration ${\cal F}^0_t:=\sigma(S_s,s\le t\wedge T)$, $t\ge 0$. Furthermore, let $(\mathcal{F}_t)$ be the right-continuous version of the raw filtration $({\cal F}^0_t)$, defined by ${\cal F}_t:=\bigcap_{s>t} {\cal F}^0_s$ for all $t\in [0,T]$. Denote by $\mathcal{M}(\Omega)$ the set of all Borel probability measures $Q$ on $\Omega$ such that the canonical process $S$ is a $Q$-martingale, and by $\mathcal{M}_c(\Omega):=\{Q\in\mathcal{M}(\Omega) : Q(K)=1 \text{ for some compact } K\subset\Omega\}$ the subset of all martingale measures with compact support.
Define
\[
  C_{\delta\sigma}:=\Big\{ X\colon \Omega\to[-\infty,+\infty] :
  \begin{array}{l} X=\liminf_n X_n
  \text{ for a sequence } (X_n) \text{ such}\\ \text{that }X_n\colon\Omega\to\mathbb{R}
  \text{ is bounded and continuous} 
  \end{array}\Big\}. 
\]
Note that $C_{\delta\sigma}$ contains all upper and lower semicontinuous functions from $\Omega$ to $\mathbb{R}$.

A process $H\colon [0,T]\times\Omega\to\mathbb{R}^d$ is called simple predictable if it is of the form
\[ 
  H_t(\omega)=\sum_{n=1}^{N} h_n(\omega)\1_{(\tau_n(\omega),\tau_{n+1}(\omega)]}(t) ,\quad (t,\omega)\in [0,T]\times \Omega,
\]
where $N\in\mathbb{N}$, $0\leq \tau_1\leq\dots\leq \tau_{N+1}\leq T$ are stopping times w.r.t.~the filtration $(\mathcal{F}_t)$, and $h_n\colon\Omega\to\mathbb{R}^d$ are bounded $\mathcal{F}_{\tau_n}$-measurable functions. The set of all simple predictable processes is denoted by $\mathcal{H}^f:=\mathcal{H}^f(\Omega)$. For a simple predictable $H\in\mathcal{H}^f$ the pathwise stochastic integral
\[ 
  (H\cdot S)_t(\omega):=\sum_{n=1}^{N} h_n(\omega) (S_{\tau_{n+1}(\omega)\wedge t}(\omega)-S_{\tau_{n}(\omega)\wedge t}(\omega)) 
\]
is well-defined for all $t\in[0,T]$ and all $\omega \in \Omega$. Similarly, the pathwise stochastic integral $(H\cdot S)$ is also well-defined for every $H\colon[0,T]\times\Omega\to\mathbb{R}^d$ in the set~$\mathcal{H}:=\mathcal{H}(\Omega)$ of processes of the form
\begin{equation*}
  H_t(\omega) = \sum_{n= 1}^{\infty} h_n(\omega) \1_{(\tau_n(\omega),\tau_{n+1}(\omega)]}(t)
\end{equation*}
where $0\leq \tau_1 \leq \tau_2\leq \cdots$ are stopping times such that for each $\omega \in \Omega$ there exists an $N(\omega)\in \N$ with $\tau_k(\omega)=T$ for all $k\geq N(\omega)$, and $h_n\colon \Omega \rightarrow \R$ are bounded $\mathcal{F}_{\tau_n}$-measurable functions. \smallskip
 
We introduce the following two assumptions, which we shall use frequently.

\begin{itemize}
  \item[(A1)] $\Omega$ is $\sigma$-compact, the metric on $\Omega$ induces a topology finer than (or equal to) 
    the one induced by the maximum norm $\|\omega\|_\infty:=\max_{t\in[0,T]} |\omega(t)|$, and for each Borel probability~$Q$ on $\Omega$ and every bounded $\mathcal{F}_t^0$-measurable function $h$ there exists a sequence of $\mathcal{F}_t^0$-measurable continuous functions $(h_n)$ which converges $Q$-almost surely to $h$.
  \item[(A2)] For every $\omega\in\Omega$ and each $t\in[0,T]$ the stopped path $\omega^t(\cdot):=\omega(\cdot\wedge t)$ is in $\Omega$ and the function $[0,T]\times\Omega \ni(t,\omega)\mapsto \omega^t$ is continuous.
\end{itemize}

If $\Omega$ is a $\sigma$-compact space endowed with the topology induced by the maximum norm, then (A1) is always satisfied, see Remark~\ref{rem:approx}. Now we are ready to state the main results of this paper. The proofs are given in Section~\ref{sec:proof}. 

\begin{theorem}\label{thm:main.integrals}
  Suppose that (A1) and (A2) hold and let $Z\colon\Omega\to[0,+\infty)$ be a continuous function 
  such that $Z(\omega^s)\leq Z(\omega^t)$ for all $\omega\in\Omega$ and $0\leq s\leq t\leq T$.
  Then, for every $X\in C_{\delta\sigma}$ which satisfies
  $X(\omega)\geq -Z(\omega)$ for all $\omega\in\Omega$, one has
  \begin{equation}\label{eq:duality}
    \inf\left\{ \lambda\in\mathbb{R} \,:\, 
    \begin{array}{l}
      \text{there is a sequence $(H^n)$ in $\mathcal{H}^f$ such that}\\
       \text{$\lambda+(H^n\cdot S)_t(\omega)\geq -Z(\omega^t) \text{ for all } (t,\omega)\in[0,T]\times\Omega$}\\
     \text{and } \lambda+ \liminf_n (H^n\cdot S)_T(\omega) \geq X(\omega) \,\text{ for all } \omega\in\Omega
     \end{array}
    \right\}
    =\sup_{Q\in\mathcal{M}_c(\Omega)} E_Q[X].
  \end{equation}
   Moreover, the equality \eqref{eq:duality} also holds 
   if $\mathcal{H}^f$ is replaced by $\mathcal{H}$, or 
   $\mathcal{M}_c(\Omega)$ is replaced by $\mathcal{M}_Z(\Omega):=\{Q\in\mathcal{M}(\Omega): E_Q[Z]<+\infty\}$.
\end{theorem}

\begin{remark}
  (i) By continuity of $Z$ one has $\mathcal{M}_c(\Omega)\subset\mathcal{M}_Z(\Omega)$. In particular, if $X(\omega)\geq- Z(\omega)$ for all $\omega\in\Omega$, the expectation $E_Q[X]$ is well-defined under every $Q\in\mathcal{M}_Z(\Omega)$. 

  (ii) Note that $Z(\omega):=\max_{t\in[0,T]} |\omega(t)|^p$ for $p\geq 0$ satisfies $Z(\omega^s)\leq Z(\omega^t)$ for every $\omega\in\Omega$ and $0\leq s\leq t\leq T$.
  
  (iii) If $Z\geq \|\cdot\|_\infty$, then $E_Q[\max_{t\in[0,T]} |S_t|]<+\infty$ for every Borel probability measure $Q$ which integrates $Z$. Hence, the set of all local martingale measures which integrate $Z$ coincides with $\mathcal{M}_Z(\Omega)$. 
\end{remark}

In particular, for $Z=0$ the previous theorem reads as follows.

\begin{corollary}\label{cor:main.integrals}
  Suppose that (A1) and (A2) hold. Then, for every $X\in C_{\delta\sigma}$ with $X\ge 0$ one has
  \[
    \inf\left\{ \lambda\in\mathbb{R} \,:\, 
    \begin{array}{l}
    \text{there is a sequence $(H^n)$ in $\mathcal{H}$ such that}\\
    \text{$\lambda+(H^n\cdot S)_t(\omega)\geq 0 \text{ for all } (t,\omega)\in[0,T]\times\Omega$}\\
    \text{and } \lambda+ \liminf_n (H^n\cdot S)_T(\omega) \geq X(\omega) \,\text{ for all } \omega\in\Omega
    \end{array}
    \right\}
    =\sup_{Q\in\mathcal{M}(\Omega)} E_Q[X].
  \]
\end{corollary}

The arguments in the proof of Theorem~\ref{thm:main.integrals} in combination with a regularity result for martingale measures on $C([0,T],\mathbb{R}^d)$ (see Lemma~\ref{lem:compact.supported.mm.are.dense} below) yields the following pricing-hedging duality on the entire space $C([0,T],\mathbb{R})$.

\begin{corollary}\label{cor:omega.is.whole.space}
  Let $\Omega=C([0,T],\mathbb{R}^d)$. Then 
  \begin{equation*}
    \inf\left\{ \lambda\in\mathbb{R} \,:\, 
    \begin{array}{l}
    \text{for every $K\subset \Omega$ compact there is $H\in\mathcal{H}^f$ and $c\geq0$}\\ 
    \text{such that $\lambda+(H\cdot S)_T(\omega)\geq -c$ for all $\omega\in \Omega$ and }\\
    \text{$\lambda+(H\cdot S)_T(\omega)\geq X(\omega)$ for all $\omega\in K$}
    \end{array}
    \right\}
    =\sup_{Q\in\mathcal{M}(\Omega)} E_Q[X]
  \end{equation*}
  for every bounded upper semicontinuous function $X\colon\Omega\to\mathbb{R}$.
\end{corollary}

\begin{remark}
  Let $\Omega \subset C([0,T],\mathbb{R})$ be a $\sigma$-compact set and let $(\pi_n)_{n\in \mathbb{N}}$ be a refining sequence of partitions of $[0,T]$ with mesh converging to zero. The pathwise quadratic variation of a path $\omega \in \Omega$ is defined by
  \begin{equation}\label{eq:quadratic varition}
    \langle\omega\rangle_t:=\liminf_n\langle\omega\rangle_t^n
    \quad \text{where}\quad
    \langle \omega\rangle^n_t:=\sum_{[u,v]\in\pi_n} (\omega(u\wedge t)-\omega(v\wedge t))^2,
  \end{equation}
  for $t\in [0,T]$.
  Then, for every continuous function $\xi\colon  \Omega \times \mathbb{R} \to\mathbb{R}$ which is bounded from below, one has
  \begin{equation*}
    X(\omega) :=\liminf_n \xi(\omega,\langle\omega\rangle^n_T)
    \in C_{\delta\sigma}.
  \end{equation*}
  Hence, the pathwise pricing-hedging duality in Theorem~\ref{thm:main.integrals} holds for this claim. This shows that the class $C_{\delta\sigma}$ includes, in  particular, the financial derivatives in the scope of~\cite{Beiglbock2017}, i.e. options on the realized variance among many others. 
\end{remark}

\begin{remark}\label{rem:counterexample}
  While the pathwise pricing-hedging duality results in \cite{Dolinsky2014,Hou2015} hold for sufficiently regular claims when trading is limited to simple predictable processes (i.e. without the ``$\liminf$'' as in our definition), the following example shows the necessity of ``$\liminf$'' for claims in $C_{\delta\sigma}$. Let $\Omega$ be the set of all H\"older continuous functions starting at zero with values in $[0,1]$. There exists a refining deterministic sequence $(\pi_n)_{n\in \mathbb{N}}$ of partitions with mesh size going to zero and a function $\tilde\omega\in \Omega$ such that
  \begin{itemize}
    \item $0\leq \tilde\omega(t) \leq 1$ for all $t\in [0,T]$, 
    \item $\langle \tilde\omega \rangle_t:=\lim_n \langle \tilde\omega \rangle^n_t$ exists for all $t\in [0,T)$ and $\lim_{t\to T}\langle \tilde\omega \rangle_t=+\infty$, 
  \end{itemize}
  where $\langle \tilde\omega \rangle^n_t$ is defined as in~\eqref{eq:quadratic varition}. For the existence of such a function~$\tilde\omega$ we refer to Lemma~\ref{lem:example function}. We fix now the above sequence $(\pi_n)_{n\in \mathbb{N}}$ and denote by $\langle \omega \rangle_t$ the corresponding quadratic variation along $(\pi_n)_{n\in \mathbb{N}}$ defined as in~\eqref{eq:quadratic varition} for all $\omega \in \Omega$. Furthermore, let us consider the option $X (\cdot):=\langle\, \cdot\,\rangle_T\in C_{\delta\sigma}$.

  Firstly, we get by It\^o's formula and Fatou's lemma that 
  \begin{equation*}
    \sup_{Q\in\mathcal{M}(\Omega)} E_Q[X]\leq 1.
  \end{equation*}
  Secondly, we observe that
  \begin{equation}\label{eq:infinity superhedging}
    \inf\{ \lambda\geq 0 : \text{there is }H\in\mathcal{H}^f \text{ such that } \lambda + (H\cdot S)_T(\omega)\geq X(\omega)\text{ for all  }\omega\in\Omega\}
    =+\infty. 
  \end{equation}
  Indeed, assume that there exists (even more general) a predicable process~$H$ of bounded variation and a constant $\lambda_0>0$ such that
  \begin{equation}\label{eq:super-hedging example}
    \lambda_0 + (H\cdot S)_T(\omega)\geq X(\omega)\quad\text{for all}\quad \omega\in \Omega
  \end{equation}
  where $(H\cdot S)_T(\omega)$ denotes the classical Riemann-Stieltjes integral defined using the integration by parts formula. For $\tilde \omega$ we get 
  \begin{equation*}
    (H\cdot S)_T (\tilde \omega)\leq \|\tilde \omega\|_{\infty} \|H(\tilde \omega)\|_{1\textup{-var};[0,T]}\leq \|H(\tilde \omega)\|_{1\textup{-var};[0,T]}<+\infty
  \end{equation*}
  where $\|H(\tilde \omega)\|_{1\textup{-var};[0,T]}$ denotes the bounded variation semi-norm of $H$, which is conflict to~\eqref{eq:super-hedging example}, i.e., that implies~\eqref{eq:infinity superhedging}.  

  Hence, there exists a duality gap if the superhedging is restricted to trading strategies of bounded variation as in \cite{Dolinsky2014,Hou2015} but the pricing-hedging duality using the limit inferior of simple predictable processes holds true since $\Omega$ and $X$ satisfy all assumptions of Theorem~\ref{thm:main.integrals}, see Section~\ref{sec:discussion} below.
\end{remark}

If $\Omega$ does not contain all its stopped paths, then the following version of Theorem~\ref{thm:main.integrals} holds true.

\begin{theorem}\label{thm:main.integrals.Z}
  Let $Z\colon\Omega\to[1,+\infty)$ be a function with compact sublevel sets $\{Z\leq c\}$ for all $c\in\mathbb{R}$ such that $Z(\omega)\geq \|\omega\|_\infty$ for all $\omega\in\Omega$. If (A1) holds true and $\mathcal{M}_Z(\Omega)\neq\emptyset$, then
  \begin{equation*}
    \inf\left\{ \lambda\in\mathbb{R} \,:\, 
    \begin{array}{l}
    \text{there is $c\ge 0$ and a sequence $(H^n)$ in $\mathcal{H}^f$ such that}\\ 
    \text{$(H^n\cdot S)_T(\omega)\geq -cZ(\omega)$ for all $\omega\in\Omega$ and} \\
    \lambda+ \liminf_n (H^n\cdot S)_T(\omega) \geq X(\omega) \,\text{ for all } \omega\in\Omega
    \end{array}
    \right\}
    =\sup_{Q\in\mathcal{M}_Z(\Omega)} E_Q[X]
  \end{equation*}
  for every $X\in C_{\delta\sigma}$ which is bounded from below.
\end{theorem}

\subsection{Relation to Vovk's outer measure}\label{subsec:vovk}

In recent years (see e.g.~\cite{Vovk2012,Vovk2016} and the references therein), Vovk introduced an outer measure on different path spaces, defined as the minimal superhedging price, which allows to quantify the path behavior of ``typical price paths'' in frictionless financial markets without any reference measure. \smallskip

In order to recall Vovk's outer measure on a set $\Omega \subset C([0,T],\R^d)$ endowed with the maximum norm, we write $\mathcal{H}_\lambda$ for the set of $\lambda$-admissible simple predicable strategies, i.e.~the set of all $H\in\mathcal{H}$ such that $(H\cdot S)_t(\omega) \ge - \lambda$ for all $(t,\omega)\in [0,T]\times\Omega$. Furthermore, we define the set of processes
\begin{align*}
  \mathcal{V}_\lambda := \left\{ \mathrm{H}:= \big(H^k\big)_{k\in \N}\,:\, H^k \in \mathcal{H}_{\lambda_k},\, \lambda_k > 0,\, \sum_{k=1}^{\infty} \lambda_k = \lambda\right\}
\end{align*}
for an initial capital $\lambda \in (0, +\infty)$. Note that for every $\mathrm{H} = \big(H^k\big) \in \mathcal{V}_{\lambda}$, all $\omega \in \Omega$, and all $t \in [0,T]$, the corresponding capital process
\begin{align*}
  (\mathrm{H}\cdot S)_t(\omega) := \sum_{k = 1}^{\infty} (H^k\cdot S)_t(\omega) = \sum_{k = 1}^{\infty} \big(\lambda_k + (H^k\cdot S)_t(\omega)\big) - \lambda
\end{align*}
is well-defined and takes values in $[-\lambda, +\infty]$. Then, Vovk's outer measure on $\Omega$ is given by 
\begin{align*}
  \overline{Q}_\Omega(A) := \inf\left\{\lambda > 0\,:\text{there is } \mathrm{H} \in \mathcal{V}_\lambda \text{ such that } 
  \lambda + (\mathrm{H}\cdot S)_T(\omega) \ge \1_A(\omega) \text{ for all } \omega \in \Omega\right\}.
\end{align*}
A slight modification of the outer measure~$\overline{Q}_\Omega$ was introduced in \citet{Perkowski2016}, which is defined as
\begin{align*}
  \overline{P}_\Omega(A):= \inf\left\{\lambda > 0:\begin{array}{l} \text{there is } (H^n) \text{ in } \mathcal{H}_\lambda \text{ such that} \\
  \lambda + \liminf_{n\rightarrow\infty}(H^n\cdot S)_T (\omega) \geq \1_A(\omega)\text{ for all } \omega \in \Omega
  \end{array} \right\}
\end{align*}
for $A \subseteq \Omega$. The latter definition seems to be more in the spirit of superhedging prices in semimartingale models as discussed in \cite[Sections~2.1 and 2.2]{Perkowski2016}. Notice that, even if it would be convenient to just minimize over simple strategies rather than over the limit (inferior) along sequences of simple strategies in both definitions of outer measures, this is essential to obtain the desired countable subadditivity of both outer measures.

\begin{remark} 
  In case that $\Omega=C([0,T],\R^d)$ one would expect that the outer measures~$\overline{Q}_\Omega$ and~$\overline{P}_\Omega$ coincide. However, currently it is only known that 
  \begin{equation}\label{eq:known estimate}
    \sup_{Q\in\mathcal{M}(\Omega)}Q(A) \leq \overline{P}_\Omega(A)\leq \overline{Q}_\Omega(A),
  \end{equation}
  where $A\subset C([0,T],\R^d)$ is a measurable set, see \cite[Lemma~6.2]{Vovk2012} and \cite[Lemma~2.9]{Perkowski2016}. In the special case of $\Omega=C([0,+\infty),\R)$ and a time-superinvariant set $A\subset C([0,+\infty),\R)$ the inequalities in \eqref{eq:known estimate} turn out to be true equalities. See \citet[Sections~2 and~3]{Vovk2012} and \citet[Section~2]{Beiglbock2017} for the precise definitions and statements in this context.
\end{remark}

By restricting the outer measure $\overline{P}_\Omega$ to a $\sigma$-compact space $\Omega$, we get the following duality result for the slightly modified version of Vovk's outer measure as a direct application of Theorem~\ref{thm:main.integrals}.

\begin{proposition}
  Under the assumptions on $\Omega$ of Theorem~\ref{thm:main.integrals}, one has
  \begin{equation*}
    \overline{P}_\Omega(A)=\sup_{Q\in\mathcal{M}(\Omega)} Q(A)
  \end{equation*}
  for all closed subsets $A\subset \Omega$.
\end{proposition}

\begin{proof}
  For every closed subset $A\subset \Omega$, it follows from Corollary~\ref{cor:main.integrals} that
  \begin{align*}
    \overline{P}_\Omega(A)&=\inf\left\{ \lambda>0 \,:\, 
    \begin{array}{l}
    \text{there is a sequence $(H^n)$ in $\mathcal{H}$ such that}\\ \text{$\lambda+(H^n\cdot S)_t(\omega)\geq 0 \text{ for all } (t,\omega)\in[0,T]\times\Omega$ and }\\ \lambda+ \liminf_n (H^n\cdot S)_T(\omega) \geq \1_A(\omega) \,\text{ for all } \omega\in\Omega
    \end{array}
    \right\}\\
    &=\sup_{Q\in\mathcal{M}(\Omega)} E_Q[\1_A]=\sup_{Q\in\mathcal{M}(\Omega)} Q(A)
  \end{align*}
  because $\1_A$ is upper semicontinuous.
\end{proof}

\begin{remark}
  Recently, \citet{Vovk2016} obtained a similar duality for open sets by adjusting the definition of the outer measure~$\overline{P}_\Omega$. More precisely, his new definition of outer measure allows for superhedging with all processes in the ``liminf-closure'' of capital processes generated by sequences of $\lambda$-admissible simple strategies, see \cite[Section~2 and Theorem~2]{Vovk2016} for more details.
\end{remark}

\subsection{Semi-static superhedging}

Let us fix a continuous function $Z\colon\Omega\to[1,+\infty)$ such that $Z(\omega^s)\leq Z(\omega^t)$ for all $\omega\in\Omega$ and $0\leq s\leq t\leq T$, and consider a finite number of securities with (discounted) continuous payoffs $G_1,\dots, G_K$ such that $|G_i|\leq cZ$ for $i=1,\dots,K$ and some $c\ge 0$. We assume that these securities can be bought and sold at prices $g_k\in\mathbb{R}$, and satisfy the no-arbitrage condition
\begin{equation}\label{NA}
  (g_1,\dots,g_K)\in\mathop{\textup{ri}}\left\{(E_Q[G_1],\dots,E_Q[G_K]):Q\in\mathcal{M}_c(\Omega)\right\}
\end{equation}
where $\mathop{\textup{ri}}$ denotes the relative interior.

Then the following semi-static hedging duality holds.

\begin{proposition}
  Suppose that the assumptions (A1) and (A2) are satisfied, and the securities with payoffs $G_1,\dots,G_K$ satisfy the static no arbitrage condition~\eqref{NA}. Then, for every upper semicontinuous function $X\colon\Omega\to\mathbb{R}$
  which satisfies $|X|\leq cZ$ for some $c\ge 0$, one has
  \begin{align}\label{eq:sem-static}
    &\inf\left\{ \lambda\in\mathbb{R} \,:\, 
    \begin{array}{l}
    \text{there is $c\ge 0$, $\alpha\in\mathbb{R}^K$, and a sequence $(H^n)$ in $\mathcal{H}^f$ such that}\\ \text{$\lambda+(H^n\cdot S)_t(\omega)\geq -cZ(\omega^t) \text{ for all } (t,\omega)\in[0,T]\times\Omega$ and }\\ \lambda+ \sum_{k=1}^K \alpha_k (G_k(\omega)-g_k)+\liminf_n (H^n\cdot S)_T(\omega) \geq X(\omega) \,\text{ for all } \omega\in\Omega
    \end{array}\right\} \nonumber\\
    &=\sup_{Q\in\mathcal{M}_c^G(\Omega)} E_Q[X]
  \end{align}
  where $\mathcal{M}^G_c(\Omega):=\{Q\in\mathcal{M}_c(\Omega): E_Q[G_k] = g_k\mbox{ for all }k=1,\dots,K\}$.
\end{proposition}    

\begin{proof}
  For every $Y\colon\Omega\to\mathbb{R}$ which satisfies $|Y|\leq cZ$ for some $c\ge 0$ we define
  \[
    \phi(Y):=\inf\left\{ \lambda\in\mathbb{R} \,:\,\begin{array}{l}
    \text{there is $c\ge 0$ and a sequence $(H^n)$ in $\mathcal{H}^f$ such that}\\ \text{$\lambda+(H^n\cdot S)_t(\omega)\geq -cZ(\omega^t) \text{ for all } (t,\omega)\in[0,T]\times\Omega$ and }\\ \lambda+ \liminf_n (H^n\cdot S)_T(\omega) \geq Y(\omega) \,\text{ for all } \omega\in\Omega
    \end{array}\right\} 
  \]
  and we remark that, by interchanging two infima, the left hand side of~\eqref{eq:sem-static} can be expressed as $\inf_{\alpha\in\mathbb{R}^K} \phi\big(X-\sum_{k=1}^K \alpha_k (G_k-g_k)\big)$. Further, Theorem~\ref{thm:main.integrals} yields
  \[
    \phi\Big(X-\sum_{k=1}^K \alpha_k (G_k-g_k)\Big)
    =\sup_{Q\in\mathcal{M}_c(\Omega)} E_Q\Big[X-\sum_{k=1}^K \alpha_k (G_k-g_k)\Big]
  \]
  for every $\alpha\in\mathbb{R}^K$. Now define the function 
  \[
    J\colon \mathcal{M}_c(\Omega)\times\mathbb{R}^K\to\mathbb{R},\quad J(Q,\alpha):=E_Q[X]-\sum_{k=1}^K \alpha_k E_Q[G_k-g_k].
  \] 
  It is immediate that $J(Q,\cdot)$ is convex for every $Q\in\mathcal{M}_c(\Omega)$ and that $J(\cdot,\alpha)$ is concave for each $\alpha\in\mathbb{R}^K$ since $\mathcal{M}_c(\Omega)$ is convex. Therefore, it follows exactly as in step (a) of the proof of \cite[Theorem~2.1]{bartl2016exponential}, that the assumption of $0$ being in the relative interior of
  \[
    \{(E_Q[G_1-g_1],\dots,E_Q[G_K-g_K]):Q\in\mathcal{M}_c(\Omega)\}
  \]
  can be used to show that all requirements of the minimax theorem~\cite[Theorem~4.1]{sion1958general} are satisfied. Hence, one gets
  \begin{align*}
    \inf_{\alpha\in\mathbb{R}^K} \phi\Big(X-\sum_{k=1}^K \alpha_k (G_k-g_k)\Big)
    & =\inf_{\alpha\in\mathbb{R}^K} \sup_{Q\in\mathcal{M}_c(\Omega)} J(Q,\alpha)\\
    & =\sup_{Q\in\mathcal{M}_c(\Omega)}\inf_{\alpha\in\mathbb{R}^K} J(Q,\alpha)
    =\sup_{Q\in\mathcal{M}_c^G(\Omega)} E_Q[X],
  \end{align*}
  where the first equality follows from Theorem~\ref{thm:main.integrals} and the last one by
  \[
    \inf_{\alpha\in\mathbb{R}^K} J(Q,\alpha)=
    \begin{cases}
    E_Q[X],&\text{if }Q\in\mathcal{M}_c^G(\Omega),\\
    -\infty,&\text{if } Q\in\mathcal{M}_c(\Omega)\setminus\mathcal{M}_c^G(\Omega).
    \end{cases}
  \]
  The proof is complete.
\end{proof}

\section{Discussion of \boldmath$\sigma$-compact spaces}\label{sec:discussion}

By definition, the $\sigma$-compactness of the metric space $\Omega\subset C([0,T],\R^d)$ requires to find a covering of $\Omega$ by compact sets $K^m$, $m\in\N$. It is an easy consequence of the Arzel{\`a}-Ascoli theorem (see e.g.~\cite[Theorem~1.4]{Friz2010}) that these $K^m$ have to be bounded, closed and equicontinuous. 

In the next lemma we provide an easy-to-check criterion for a set $\Omega$ of continuous functions to be $\sigma$-compact. This leads to many interesting examples of such $\Omega\subset C([0,T],\mathbb{R}^d)$ appearing in the context of (classical) financial modeling, see Subsection~\ref{subsec:examples}.

\begin{lemma}\label{lem:simga compact}
  For $n\in \mathbb{N}$ let $c_n\colon [0,T]^2\to [0,+\infty)$ be a continuous function with $c_n(t,t)=0$ for $t\in [0,T]$ and define the norm
  \begin{equation*}
    \|\omega\|_{c_n,\alpha}:= |\omega(0)| + \sup_{s,t\in[0,T]} \frac{|\omega (t)-\omega(s)|}{c_n(s,t)^\alpha},\quad \omega \in C([0,T],\mathbb{R}^d),
  \end{equation*}
  with $\alpha \in (0,1]$ and the convention $\frac{0}{0}:=0$.
  Then the spaces 
  \begin{equation*}
    \Omega_n := \left\{ \omega \in C([0,T],\mathbb{R}^d)\,:\, \|\omega\|_{c_n,1}<+\infty \right\},\quad n\in \N,
  \end{equation*}
  are $\sigma$-compact w.r.t. the norm $\|\cdot\|_{c_n,\alpha}$ for $\alpha \in (0,1)$ and in particular w.r.t.~the maximum norm $\|\cdot\|_\infty$. 
  Moreover, the set $\Omega := \bigcup_{n\in \N} \Omega_n$ is $\sigma$-compact w.r.t.~the maximum norm $\|\cdot\|_\infty$.   
\end{lemma}

\begin{proof}
  For $m,n\in\N$ we observe   
  \begin{equation*}
    \Omega_n = \bigcup_{m\in \N} K^m_n \quad \text{with}\quad K_n^m := \left\{ \omega \in C([0,T],\mathbb{R}^d)\,:\, \|\omega\|_{c_n,1}\leq m \right\}.
  \end{equation*}
  In order to show the $\sigma$-compactness of $\Omega_n$ w.r.t. $\|\cdot\|_{\infty}$, we need to show that each $K^m_n$ is compact. Due to the Arzel{\`a}-Ascoli theorem, it is sufficient to show that each~$K^m_n$ is bounded, equicontinuous and closed. 
  
  \noindent Boundedness: For every $\omega\in K^m_n$ we have 
  \begin{equation*}
    \|\omega\|_\infty \leq |\omega(0)| + \sup_{t\in [0,T]}|\omega(t)-\omega(0)|\leq |\omega(0)| + m \sup_{t\in [0,T]} c_n(0,t).
  \end{equation*}
 
  \noindent Equicontinuity: Because $c_n$ is continuous on a compact set and $c_n(t,t)=0$ for $t\in [0,T]$, for every $\epsilon >0$ there exits a $\delta >0$ such that $|c_n(s,t)| <\epsilon/m$ for $|t-s|\leq \delta$. Hence, for every $\omega \in K^m_n$ and $s,t\in [0,T]$ with $|t-s|\leq \delta$ we get $|\omega (t)-\omega (s)| \leq \epsilon$.   

  \noindent Closeness: If $(\omega_k)\subset K^m_n$ converges uniformly to $\omega$, then $\omega \in K^m_n$. Indeed, this can be seen by 
  \begin{equation*}
    |\omega (0)|+\frac{|\omega (t)-\omega (s)|}{c_n(s,t)}= \lim_{k\to \infty} \bigg( |\omega_k (0)|+ \frac{|\omega_k (t)-\omega_k (s)|}{c_n(s,t)} \bigg)\leq m.
  \end{equation*}
 
  The $\sigma$-compactness of $\Omega_n$ w.r.t.~$\|\cdot\|_{c_n,\alpha}$ for $\alpha \in (0,1)$ follows by the fact that the uniform convergence in each $K^m_n$ implies the convergence w.r.t. $\|\cdot\|_{c_n,\alpha}$, which is a consequence of the following interpolation inequality
  \begin{equation*}
    \frac{|\omega (t)-\omega(s)|}{c_n(s,t)^{\alpha}}  
    =\bigg( \frac{|\omega (t)-\omega(s)|}{c_n(s,t)} \bigg)^{\alpha}|\omega (t)-\omega(s)|^{1-\alpha}
    \leq 2 \|\omega\|_{c_n,1}^\alpha \|\omega\|_\infty^{1-\alpha},\quad s,t\in [0,T].
  \end{equation*}
  
  Finally, $\Omega$ is $\sigma$-compact (w.r.t.~$\|\cdot\|_{\infty}$) since it is a countable union of $\sigma$-compact sets.
\end{proof}

From the previous lemma it is easy to deduce that many well-known function spaces $\Omega \subset C([0,T],\R^d)$ are $\sigma$-compact spaces. To state the next corollary, we recall the notion of control functions: $c\colon [0,T]^2\to [0,+\infty)$ is called control function if $c$ is continuous, super-additive, i.e.~$c(s,t)+c(t,u)\leq c(s,u)$ for $0\leq s\leq t\leq u\leq T$, and $c(t,t)=0$ for every $t\in [0,T]$.
 
\begin{corollary}\label{cor:sigma compact examples}~
  (i) The space $C^{\alpha}([0,T],\R^d)$ of $\alpha$-H\"older continuous functions, i.e.
  \begin{equation*}
    C^{\alpha}([0,T],\R^d) := \left\{ \omega \in  C([0,T],\R^d) \,:\, \sup_{s,t\in [0,T]}\frac{|\omega(t)-\omega(s)|}{|t-s|^\alpha} <+\infty\right \}, \quad \alpha \in (0,1],
  \end{equation*}
  is $\sigma$-compact w.r.t.~$\|\cdot\|_\infty$ and w.r.t.~the  H\"older norm $\|\cdot\|_{\beta}$ for $\beta \in (0,\alpha)$ defined by
  \begin{equation*}
    \|\omega\|_{\beta}:= |\omega(0)|+ \sup_{s,t\in [0,T]}\frac{|\omega(t)-\omega(s)|}{|t-s|^\beta}  \quad \text{for}\quad \omega \in C^{\alpha}([0,T],\R^d).
  \end{equation*}
  
  (ii) The space $C^{\textup{H\"older}}([0,T],\R^d):= \bigcup_{\alpha \in (0,1]} C^{\alpha}([0,T],\R^d)$ of all H\"older continuous functions is $\sigma$-compact w.r.t.~the maximum norm $\|\cdot\|_\infty$. 
  
  (iii) The fractional Sobolev space $W^{\delta,p}([0,T],\R^d)$ with $\delta -1/p>0$, given by
  \begin{equation*}
    W^{\delta,p}([0,T],\R^d):=  \left\{ \omega \in  C([0,T],\R^d) \,:\, \int_{[0,T]^2}\frac{|\omega(t)-\omega(s)|^p}{|t-s|^{\delta p +1}}\dd s \dd t <+\infty\right \}
  \end{equation*}  
  for $\delta \in (0,1)$ and $ p\in [1,+\infty)$, is $\sigma$-compact w.r.t.~maximum norm $\|\cdot\|_\infty$. 
  
  (iv) The space $C^{p\var,c}([0,T],\R^d)$, which is a subspace of continuous functions with finite $p$-variation, given by
  \begin{equation*}
    C^{p\var,c}([0,T],\R^d) := \left\{ \omega \in  C([0,T],\R^d) \,:\, \sup_{s,t\in [0,T]}\frac{|\omega(t)-\omega(s)|}{c(s,t)^{1/p}} <+\infty\right \}
  \end{equation*}
  for $p\in [1,+\infty)$ and a control function $c$, is $\sigma$-compact w.r.t.~the maximum norm $\|\cdot\|_\infty$ and w.r.t.~the $p^\prime$-variation norm $\|\cdot\|_{p^\prime\var}$ for $p^\prime \in (p,+\infty)$ defined by
  \begin{equation*}
    \|\omega\|_{p^\prime \var} := |\omega(0)| + \sup_{0\leq t_0\leq \cdots \leq t_n \leq T,\, n\in \N} \bigg ( \sum_{i=0}^{n-1} |\omega(t_{i+1})-\omega(t_i)|^{p^\prime} \bigg)^{1/p^\prime}.
  \end{equation*}
\end{corollary}

\begin{proof}
  \textit{(i) and (ii)} follow directly by Lemma~\ref{lem:simga compact} and the fact that 
  \begin{equation*}
    C^{\alpha}([0,T],\R^d) \subset C^{\frac{1}{n}}([0,T],\R^d)\quad \text{for}\quad \alpha \in [n^{-1},(n-1)^{-1}], \quad n\in \N.
  \end{equation*}
 
  \textit{(iii)} Classical Sobolev embedding results, see e.g. \cite[Corollary~A.2]{Friz2010}, imply
  \begin{equation*}
    W^{\delta,p}([0,T],\R^d) \subset C^{\delta-1/p}([0,T],\R^d)\quad \text{and}\quad \|\omega\|_{\delta-1/p}\leq C(\delta,p) \|\omega\|_{W^{\delta,p}}
  \end{equation*} 
  for $\omega \in W^{\delta,p}([0,T],\R^d)$ with $\delta - 1/p>0$ and for a constant $C(\delta,p)>0$ depending only on $\delta$ and p. Here $\|\cdot\|_{W^{\delta,p}}$ denotes the fractional Sobolev semi-norm, see~\eqref{eq:Sobolev norm} below. Hence, to obtain the stated $\sigma$-compactness from Lemma~\ref{lem:simga compact}, it remains to show that, if a sequence $(\omega_k)\subset  W^{\delta,p}([0,T],\R^d)$ with $\|\omega\|_{W^{\delta,p}}\leq K$ for some constant $K>0$ converges uniformly to a function $\omega$, then $\|\omega\|_{W^{\delta,p}}\leq K$. However, this is a simple consequence of Fatou's lemma. 
 
  \textit{(iv)} The $\sigma$-compactness w.r.t.~$\|\cdot\|_\infty$ and $\|\cdot\|_{c,\alpha}$ for $\alpha \in (0,1)$ follows again by Lemma~\ref{lem:simga compact}. The $\sigma$-compactness w.r.t.~$\|\cdot\|_{p^\prime \var}$ can be deduced from the inequality 
  $$
    \|\omega\|_{p^\prime \var} \leq \|\omega \|_{c,\frac{1}{p}} c(0,T)^{1/p}
  $$
  for $\omega \in C^{p\var,c}([0,T],\R^d)$ and for $p^\prime \in (p,+\infty)$.
\end{proof}

\begin{remark} 
  (i) The function spaces stated in Corollary~\ref{cor:sigma compact examples} satisfy also the first part of assumption~(A2): for every $\omega\in\Omega$ and $t\in[0,T]$ the stopped path $\omega^t(\cdot):=\omega(\cdot\wedge t)$ is in $\Omega$; in the case of the H\"older-type spaces this is fairly easy to verify and for the Sobolev space we refer to \cite[Lemma~1.5.1.8]{Grisvard1985}. Hence, all these function spaces equipped with the maximum norm satisfy the assumptions (A1) and (A2), see also Remark~\ref{rem:approx}.
  
  (ii) From the perspective of (completely) model-free financial mathematics it might be desirable to consider the space $C^{p\var}([0,T],\R^d)$ of all continuous functions possessing finite $p$-variation for $p>2$ since this space includes the support of all martingale measures. Unfortunately, the elementary covering as used in the proof of Lemma~\ref{lem:simga compact} cannot work as the unit ball in $C^{p\var}([0,T],\R^d)$ is not compact, see e.g.~\cite[Example~3.4]{Maligranda1992}.
\end{remark}

\subsection{Examples from mathematical finance}\label{subsec:examples}

As mentioned in the introduction, the prediction set~$\Omega$ can be interpreted to contain all the price paths that an investor believes could possibly appear at a financial market. Hence, it is natural to choose~$\Omega$ in a way that it includes those price processes coming from financial models which have been proven to provide fairly reasonable underlying price processes. 

\begin{example}
  A natural assumption coming from semi-martinagle models is to consider a prediction set $\Omega_{\textup{QV}}$ of continuous paths possessing pathwise quadratic variation in the sense of F\"ollmer. We refer, e.g., to the work~\cite{Schied2016} (and the references therein) for such frameworks. To be more precise, fix a refining sequence of partitions~$(\pi_n)_{n\in \mathbb{N}}$ with mesh size going to zero and consider the prediction set 
  \begin{equation*}
    \Omega_{\textup{QV}}:= \{\omega \in C^{\alpha}([0,T],\mathbb{R})\,:\, \omega(0)=0\text{ and } \| \omega \|_{QV} <C \}
  \end{equation*}
  for $\alpha\in (0,1)$ and some constant $C>0$, where 
  \begin{equation*}
    \| \omega \|_{\textup{QV}}:= \sup_{n\in \mathbb{N}} \bigg(\sum_{[s,t]\in \pi_n} |\omega(t)-\omega(s)|^2 \bigg)^{1/2}.
  \end{equation*}
  Note that $\Omega_{\textup{QV}}$ is $\sigma$-compact with respect to the norm~$\|\cdot\|_{\infty}$. Indeed, we have 
  \begin{equation*}
    \Omega_{\textup{QV}}= \bigcup_{n\in \mathbb{N}} \Omega_n \quad\text{and}\quad \Omega_n:=\{\omega \in C^{\alpha}([0,T],\mathbb{R})\,:\,\|\omega\|_\alpha +\|\omega\|_{\textup{QV}} \leq n \},
  \end{equation*}
  where $\Omega_n$ is a compact set for each $n\in \mathbb{N}$. In order to see the compactness of~$\Omega_n$, we observe that the condition $\|\omega\|_\alpha +\|\omega\|_{\textup{QV}}\leq n$ ensures that the set~$\Omega_n$ is equicontinuous and uniformly bounded and, furthermore, every sequence in $(\omega_m) \subset \Omega_n$ possesses a subsequence which convergences in the maximum norm to a function $\omega \in C^{\alpha}([0,T],\mathbb{R})$ with $\|\omega\|_{\alpha} \leq n$, cf. Lemma~\ref{lem:simga compact}. The required bound $\|\omega\|_{\textup{QV}}\leq n $ follows by the same estimates as used for the proof of \cite[Proposition~5.28]{Friz2010}.

  Let us consider, for instance, a simple lookback option $X(\omega):=\max_{t\in [0,T]}|\omega(t)| $ on the market $\Omega_{\textup{QV}}$. Using a pathwise version of the Buckholder-Davis-Gundy inequality (\cite[Theorem~2.1]{Beiglbock2015}), we get
  \begin{equation*}
    X(\omega) \leq \liminf_{ n } \max_{t\in \pi^n} |S_{t}(\omega)|\leq 6 \sqrt{C} + \liminf_n (H^n \cdot S)_T(\omega)
  \end{equation*}
  for all $\omega \in \Omega_{\textup{QV}}$ and for some simple predicable process $(H^n)$. From this we can conclude that the superhedging price is less or equal to $ 6 \sqrt{C}$, using the definitions from Theorem~\ref{thm:main.integrals}. Note that the superhedging price on the entire space $C([0,T],\mathbb{R})$ has to be infinity if we aim to have the duality between the superhedging price and the supremum of $E_Q[X]$ over all martinagle measures~$Q$.
\end{example}

\begin{example}
  Instead of using financial model based on semi-martingales, there is a rich literature on financial modeling using fractional Brownian motion because of its favorable time-series properties, see e.g.~\cite{Rostek2013} and the references therein.
  
  This motivates the choice of the prediction set $\Omega:= \{\omega\in C^{H}([0,T],\mathbb{R})\,:\, \omega(0)=0\}$ as it contains the sample paths of a fractional Brownian motion with Hurst index $H\in (0,1)$. In the case of $H>1/2$, for every upper semicontinuous claim $X\colon \Omega\to [0,+\infty]$, we can apply our pathwise pricing-hedging duality (Theorem~\ref{thm:main.integrals}), to see that the superhedging price is given by
  \begin{equation*}
    \phi (X) :=
    \inf\left\{ \lambda\in\mathbb{R} \,:\, 
    \begin{array}{l}
      \text{there is a sequence $(H^n)$ in $\mathcal{H}^f$ such that}\\
       \text{$\lambda+(H^n\cdot S)_t(\omega)\geq 0 \text{ for all } (t,\omega)\in[0,T]\times\Omega$}\\
     \text{and } \lambda+ \liminf_n (H^n\cdot S)_T(\omega) \geq X(\omega) \,\text{ for all } \omega\in\Omega
     \end{array}
    \right\} = X(0),
  \end{equation*}
  where $0$ stands for the constant path equal to~$0$, since the Dirac measure at~$0$ is the only martingale measure in $\mathcal{M}_c(\Omega)$. Notice that the pathwise superhedging price considering the entire space~$C([0,T],\R)$ is $\sup_{\omega \in C([0,T],\R)}|X (\omega)|$ for many options~$X$.
\end{example}

\begin{remark}
  Prediction sets can naturally be modeled by means of the pathwise quadratic variation~\eqref{eq:quadratic varition}. For instance the typical price paths of the Black-Scholes model are given by the prediction set
  \[
    \Omega=\left\{\omega\in C([0,T],\mathbb{R}): \omega(0)=s_0\mbox{ and }\langle\omega\rangle_\cdot=\int_0^\cdot \sigma^2 \omega(t)^2\dd t\right\},\quad s_0\in  \mathbb{R}.
  \]
  However, prediction sets depending on the pathwise quadratic variation as in the previous example are in general not $\sigma$-compact and the duality results of this paper do not apply. As shown in Bartl et al.~\cite{Bartl2017}, a pathwise pricing-hedging duality on such prediction sets can still be obtained but it requries a modified superhedging price which allows to invest directly in the quadratic variation: this new superhedging price of a contingent claim $X$ is defined as the inifimum over all $\lambda\in\mathbb{R}$ for which there exist sequences $(H^n)$ and $(G^n)$ of simple strategies such that
  \begin{equation*}
    \lambda+ \liminf_{n\to\infty} \big((H^n\cdot S)_T(\omega)
    +(G^n\cdot \int S\dd S)_T(\omega)\big) \ge X(\omega)\quad\mbox{for all }\omega\in\Omega
  \end{equation*}
  and the admissibility condition $\lambda+(H^n\cdot S)_t(\omega)+(G^n\cdot \int S\dd S)_t(\omega)\geq 0$ for all $n\in\mathbb{N}$, $\omega\in\Omega$, and $t\in[0,T]$. The key idea is to extend the market model and consider a two dimensional price process $(S,\int S\,dS)$ on the product space $C([0,T],\mathbb{R})\times C([0,T],\mathbb{R})$
  and adapting the duality results (and its proofs) of the present paper accordingly. For a detailed discussion on prediction sets depending the pathwise quadratic variation we refer to \cite{Bartl2017}. 
\end{remark}

In the following we present several examples coming from the modeling of financial market and guarantee the existence of $\sigma$-compact metric spaces $\Omega\subset C([0,T],\R^d)$, which include all the possible price trajectories produced by these models. For simplicity we consider one-dimensional processes and denote by~$W$ a one-dimensional Brownian motion on a probability space $(\tilde \Omega, P,\mathcal{F})$. However, all arguments extend straightforward to multi-dimensional settings. 

\begin{example}[Classical Black-Scholes model]
  A classical example from mathematical finance is the famous Black-Scholes model, which is given by
  \begin{equation*}
    \d S_t = \sigma S_t \dd W_t +\mu S_t \dd t, \quad t\in [0,T],
  \end{equation*}
  for $\mu\in \R$ and $\sigma>0$. In this case the price process $S$ is a so-called geometric Brownian motion, which possesses the same sample path regularity as a Brownian motion. Hence, one has almost surely $S\in C^{\alpha}([0,T],\R)$ and $S\in W^{\alpha-\frac{1}{q},q}([0,T],\R)$ for every $\alpha \in (0,1/2)$ and $q>2$, cf. Corollary~\ref{cor:ito processes}. 
\end{example}

\begin{example}[Local volatility models]
  Other examples are local volatility models
  \begin{equation*}
    \d S_t = \sigma (t, S_t) \dd W_t,\quad S_0=s_0,\quad t\in [0,T],
  \end{equation*}
  for a volatility function $\sigma \colon [0,T]\times \R\to \R$. For these classes of models one again has $S\in  \Omega := C^{\alpha}([0,T],\R)$ a.s. for every $\alpha <1/2$ if $s_0\in \R$ and $\sigma$ is Lipschitz continuous and satisfies the linear growth condition $|\sigma (t,x)|^2\leq K (1+|x|^2)$ for $(t,x)\in [0,T]\times \R$ and positive constant $K>0$. Indeed, the H\"older regularity of $S$ can be deduced from Corollary~\ref{cor:ito processes} combined with the estimate 
  \begin{align*}
    E_P \bigg[\int_0^T |\sigma (s,S_s)|^q \dd s\bigg]\leq \tilde C  E_P \bigg[\int_0^T (1+|S_s|)^q \dd s\bigg] 
    \leq \tilde C^\prime \bigg(1+ \int_0^T E_P  \big[|S_s|^q\big] \dd s\bigg) \leq C,
  \end{align*}
  for constants $\tilde C, \tilde C^\prime>0$ and $C=C(q,K,T,S_0)>0$, and for every $q\geq 2$, where the last inequality follows by the $L^q$-estimate in \cite[Theorem~4.1]{Mao2008}.
\end{example}

\begin{example}[Stochastic volatility models (with uncertainty)]
  A frequently used generalization of the Black-Scholes model is given by stochastic volatility models
  \begin{equation}\label{eq:stochastic vol model}
    \d S_t = \sigma_t  S_t \dd W_t + \mu_t S_t \dd t,\quad S_0=s_0, \quad t\in [0,T],
  \end{equation}
  for $s_0\in \R$ and predictable real-valued processes $\mu$ and $\sigma$. This type of linear stochastic differential equations can be explicitly solved by 
  \begin{equation*}
    S_t := s_0\exp \bigg( \int_0^t \bigg(\mu_s -\frac{\sigma_s^2}{2} \bigg)\dd s + \int_0^t \sigma_s \dd W_s \bigg),\quad t\in [0,T].
  \end{equation*}
  Based on Corollary~\ref{cor:ito processes}, one can easily deduce the sample path regularity of the price process~$S$: For $q\in (2,+\infty)$, $\alpha \in (0,1/2-1/(2q))$ and $\delta:= \alpha-1/q$, if $ E_P \big[\int_0^T|\mu_s|^{q}\dd s\big]<+\infty$ and $ E_P \big[\int_0^T|\sigma_s|^{2q}\dd s\big]<+\infty$, then 
  \begin{equation}\label{eq:regularity price process}
    S\in C^{\alpha}([0,T],\R) \quad \text{and} \quad S\in W^{\delta,q}([0,T],\R),\quad \text{a.s.}
  \end{equation}
  For example the Heston model is a stochastic volatility model, in which the volatility process~$\sigma$ satisfies such a bound.

  In the context of stochastic volatility modeling with Knightian uncertainty, one usually replaces the fixed volatility process~$\sigma$ by a class of volatility processes. For example the seminal works on volatility uncertainty~\cite{Avellaneda1995} and~\cite{Lyons1995} require the volatility processes~$\sigma$ to be such that $\sigma_t \in [\sigma_{\textup{min}},\sigma_{\textup{max}}]$ for all $t\in [0,T]$ and some constants $\sigma_{\textup{min}},\sigma_{\textup{max}}>0$ with $\sigma_{\textup{min}}<\sigma_{\textup{max}}$. Therefore, due to the bounds on the volatility, all possible price paths considered in~\cite{Avellaneda1995} and~\cite{Lyons1995} belong to the function spaces as stated in~\eqref{eq:regularity price process}.
\end{example}

\begin{example}[Rough volatility models]
  Recently, analyzing time series of volatility using high frequency data, Gatheral, Jaisson and Rosenbaum~\cite{Gatheral2014} showed that the log-volatility behaves essentially like a fractional Brownian motion with Hurst exponent $H$ close to $0.1$. This new insight has led to various fractional extensions of classical volatility models (see e.g.~\cite{Gatheral2014,Bayer2016,Bennedsen2016,ElEuch2019}), which nicely lead to price paths belonging to the $\sigma$-compact metric space of H\"older continuous functions. Indeed, if the stochastic volatility~$\sigma$ fulfills for some $M>0$ and $q>r\geq 1$ the bound
  \begin{equation}\label{eq:volatility regularity}
    E_P[ |\sigma_t-\sigma_s|^q]\leq |t-s|^{\frac{q}{r}}\quad \text{for } s,t\in [0,T]\quad \text{and}\quad  \sigma_0 \in \R, 
  \end{equation}
  then we observe that 
  \begin{equation*}
    E_P\bigg[\int_0^T |\sigma_s|^q \dd s\bigg]\leq C  \big(|\sigma_0|^q + E_P\big[\|\sigma\|_{\beta}^q\big]\big)<+\infty, 
  \end{equation*}
  for some constant $C=C(q,M,T)>0$ and $\beta \in (0,1/r-1/q)$. Note, that condition~\eqref{eq:volatility regularity} is exactly the condition usually required by the Kolmogorov continuity criterion (cf. Theorem~\ref{thm:Kolmogorov}), which is frequently used to verify the H\"older regularity of a stochastic process. In particular, every rough volatility model satisfying~\eqref{eq:volatility regularity} with associated price process given by~\eqref{eq:stochastic vol model}  generates price paths possessing H\"older regularity as provided in~\eqref{eq:regularity price process}. For example, a simple fractional Brownian motion with Hurst index~$H$ fulfills the bound~\eqref{eq:volatility regularity} with $q\in [2,+\infty)$ and $r=H$ and the rough Heston model as introduced by El Euch and Rosenbaum~\cite[(1.3)]{ElEuch2019} fulfills the bound~\eqref{eq:volatility regularity} with $q\in [2,+\infty)$ and $1/r=\alpha-1/2$ for $\alpha \in (1/2,1)$, where $\alpha$ denotes the parameter specified in the rough Heston model~\cite[(1.3)]{ElEuch2019}.
\end{example}

\begin{example}[Volatility uncertainty]
  The most general case of volatility uncertainty is usually provided by simultaneously considering all processes of the type
  \begin{equation*}
    S_t =\int_0^t \sqrt{\sigma_s} \dd W_s, \quad t\in [0,T],
  \end{equation*}
  for strictly positive and predictable processes $\sigma$, see~\cite{Neufeld2013,Possamai2013}. While they can deal with all $\sigma$ such that $\int_0^T\sigma_s\dd s<+\infty$ a.s., we have seen in 3. that we can deal with all volatility processes~$\sigma$ such that $E_P[\int_0^T\sigma_s^q\dd s]<+\infty$ for $q\in (1,+\infty)$.

  Another sub-class of these price processes~$S$ leading to $\sigma$-compact sets of price paths is given by all processes~$S$ with corresponding volatility process~$\sigma$ such that $\sigma \leq f$ for some deterministic integrable function $f\colon [0,T]\to (0,+\infty)$. Indeed, defining the quadratic variation of $S$ by $\langle S\rangle_t=\int_0^t \sigma_s \dd s$ for $t\in [0,T]$ and using Dambis Dubin-Schwarz theorem, one has $S_t = B_{\langle S\rangle_t}$ for a suitable Brownian motion $B$. Based on this observation, it is easy to derive that 
  \begin{equation*}
    S \in C^{p\var,c}([0,T],\R),\quad \text{a.s.,}\quad \text{with}\quad  c (s,t):= \int_s^t f(r)\dd r, \quad s,t\in [0,T],
  \end{equation*}
  and $p>2$. Recall that $C^{p\var,c}([0,T],\R)$ is $\sigma$-compact by Corollary~\ref{cor:sigma compact examples}.
\end{example}

\section{Proofs of the main results}\label{sec:proof}

Denote by $C_b$ the set of all bounded continuous functions $X\colon\Omega\to\mathbb{R}$.

\begin{remark}\label{rem:approx}
  If $\Omega$ is a $\sigma$-compact space endowed with the maximum norm, then (A1) is always satisfied.
\end{remark}

\begin{proof}
  Fix $t\in[0,T]$, a bounded $\mathcal{F}_t^0$-measurable function $h$, and a Borel probability $Q$. Define $\pi\colon \Omega\to C([0,t],\mathbb{R}^d)$, $\pi(\omega)(s):=\omega(s)$, and set $\Omega_t:=\pi(\Omega)$ endowed with the maximum norm $\|\omega\|_{\infty}:=\max_{s\in[0,t]}|\omega(s)|$. By $\sigma$-compactness there exist compact sets $K_n$, $n\in\mathbb{N}$, such that $\Omega=\bigcup_n K_n$. Further, since $\Omega_t=\bigcup_n \pi(K_n)$, and $\pi(K_n)$ is compact by continuity of $\pi$, it follows that $\Omega_t$ is $\sigma$-compact and therefore separable. Standard arguments show that $\mathcal{F}_t^0=\{\pi^{-1}(B): B\in\mathcal{B}(\Omega_t)\}$, where $\mathcal{B}(\Omega_t)$ denotes the Borel sets of $\Omega_t$. Hence, $h=\tilde{h}\circ\pi$ for some Borel function $\tilde{h}\colon\Omega_t\to\mathbb{R}$. Again by $\sigma$-compactness of $\Omega_t$, the probability measure $\tilde{Q}:=Q\circ\pi^{-1}$ is tight and thus regular, i.e.~Borel sets can be approximated by compact subsets in measure. In particular, there exists a sequence of continuous functions $\tilde{h}_n\colon\Omega_t\to\mathbb{R}$ such that $\tilde{h}_n\to\tilde{h}$ $\tilde{Q}$-almost surely, which in turn implies $h_n:=\tilde{h}_n\circ\pi\to \tilde{h}\circ\pi=h$ $Q$-almost surely.
\end{proof}

The following lemma is an immediate consequence of standard results about discrete-time local martingales (see \cite[Theorems~1 and~2]{jacod1998local}), which we recall for later references. 

\begin{lemma}\label{lem:weak}
  If $Q\in\mathcal{M}(\Omega)$ and $H\in\mathcal{H}^f$ such that $E_Q[(H\cdot S)_T^-]<+\infty$, then $(H\cdot S)_T$ is $Q$-integrable and $E_Q[(H\cdot S)_T]=0$.
\end{lemma}

Next we need to establish some auxiliary results. 

\begin{lemma}\label{lem:S.stopped.lsc}
  Let $d=1$, $0\leq s< t\leq T$, $m> 0$, and define 
  \[
    \tau:=\inf\{r\geq s: S_r>m \text{ or } S_r\leq -m\}\wedge T.
  \]
  Then the function $\omega\mapsto S_{\tau(\omega)\wedge t}(\omega)$ is lower semicontinuous w.r.t.~the maximum norm.
\end{lemma}

\begin{proof}
  Define $\tau_+:=\inf\{r\geq s: S_r>m \}\wedge T$ and $\tau_-:=\inf\{r\geq s: S_r\leq -m\}\wedge T$, and note that $\tau=\tau_+\wedge \tau_-$. Moreover, fix $\omega$ and a sequence $(\omega_n)$ such that $\|\omega_n-\omega\|_\infty\to 0$. We claim that
  \[ 
    \limsup_n \tau_+(\omega_n)\leq\tau_+(\omega)\quad\text{and}\quad \liminf_n \tau_-(\omega_n)\geq\tau_-(\omega).
  \]
  Indeed, assume without loss of generality that $r:=\tau_+(\omega)<T$. 
  Then, by defintion, for every $\varepsilon>0$ there is $\delta\in(0,\varepsilon)$ such that
  $\omega(r+\delta)>m$. Therefore $\omega_n(r+\delta)>m$ for eventually all $n$, 
  showing that $\tau_+(\omega_n)\leq r+\varepsilon$ for eventually all $n$. 
  As $\varepsilon$ was arbitrary, the first part of the claim follows.  
  Next, we may assume without loss of generality that $r:=\tau_-(\omega)>s$. Then necessarily $\omega(u)>-m$ for $u\in[s,r)$. By continuity of $\omega$ and since $\|\omega_n-\omega\|_\infty\to 0$, for every $\varepsilon>0$, it holds $\omega_n(u)>-m$ for all $u\in[s,r-\varepsilon]$ and therefore $\tau_-(\omega_n)\geq r-\varepsilon$ for eventually all $n$. As $\varepsilon$ was arbitrary, the second part of the claim follows. In the following we prove the lower semicontinuity of $S^\tau_t$.
	
  (a) If $S^\tau_t(\omega)>m$, then $\tau(\omega)=\tau_+(\omega)=s$ and $\omega(s)>m$. In particular $\omega_n(s)>m$ and $\tau_+(\omega_n)=s$ for eventually all $n$, hence $\lim_n S^\tau_t(\omega_n)=\lim_n \omega_n(s)=\omega(s)=S^\tau_t(\omega)$.
	
  (b) If $S^\tau_t(\omega)=m$, then either $\tau_+(\omega)<t$ or $\tau_+(\omega)\geq t$. In the first case it follows that $\tau_+(\omega)<\tau_-(\omega)$ so that $\tau_+(\omega_n)<\tau_-(\omega_n)$ and $\tau_+(\omega_n)<t$ for all but finely many $n$ by the first part of the proof and therefore
  \[
    \liminf_n S^\tau_t(\omega_n)
    =\liminf_n \omega_n(\tau_+(\omega_n))
    =m
    =S^\tau_t(\omega).
  \]
  On the other hand, if $\tau_+(\omega)\geq t$, then $\omega(t)=m$ and $\omega(r)>-m$ for $r\in[s,t]$. This implies that $\tau_-(\omega_n)\geq t$ for eventually all $n$ and therefore
  \[
    \liminf_n S^\tau_t(\omega_n)
    =\liminf_n \omega_n(t\wedge \tau_+(\omega_n))
    = m
    =S^\tau_t(\omega).
  \]
  
  (c) If $S^\tau_t(\omega)\in (-m,m)$, then either $\tau(\omega)>t$ or $\tau(\omega)=t$ (in which case necessarily $t=T$). 
  In the latter case it follows that $\omega(r)>-m$ for $r\in[s,T]$, hence $\tau_-(\omega_n)=T$ for eventually
  all $n$ and thus
  \[
    \liminf_n S^\tau_t(\omega_n)
    =\liminf_n \omega_n(t\wedge \tau_+(\omega_n))
    \geq \omega(t)
    =S^\tau_t(\omega).
  \]
  If $\tau(\omega)>t$, then again $\tau_-(\omega_n)>t$ for eventually all $n$ so that the same argument shows that $\liminf_n S^\tau_t(\omega_n)\geq S^\tau_t(\omega)$.
	
  (d) If $S^\tau_t(\omega)=-m$, then $\omega(s)\geq -m$. Assume that $\liminf_n  S^\tau_t(\omega_n)<-m$. Then there is a subsequence still denoted by $(\omega_n)$ such that $\tau(\omega_n)=\tau_-(\omega_n)=s$. However, this contradicts $\liminf_{n}S^\tau_t(\omega_n)=\lim_n \omega_n(s)=\omega(s)\ge -m$.
  
  (e) If $S^\tau_t(\omega)<-m$, then $\tau_-(\omega)=s$ and $\omega(s)<-m$. This implies $\omega_n(s)<-m$ and therefore $\tau_-(\omega_n)=s$ for eventually all $n$, so that $\lim_n S^\tau_t(\omega_n)=\lim_n\omega_n(s)=\omega(s)=S^\tau_t(\omega)$.
\end{proof}

\begin{proposition}\label{lem:polar}
  Assume that (A1) holds true. Then, for any Borel probability measure $Q$ on $\Omega$ which is not a local martingale measure, there exist $X\in C_b$ and $H\in\mathcal{H}^f$ such that $X\leq (H\cdot S)_T$ and $E_Q[X]>0$.
\end{proposition}

\begin{proof}
  Notice that $S$ is a local martingale if and only if each component is a local martingale, which means we may assume without loss of generality that $d=1$.

  We prove that if $E_Q[X]\leq 0$ for all $X\in G:=\{X\in C_b : X\leq (H\cdot S)_T \text{ for some } H\in\mathcal{H}^f\}$, then $Q$ is a local martingale measure, i.e.~for every $m\in\mathbb{N}$, the stopped process
  \[
    S^\tau_t:=S_{t\wedge \tau}\quad\text{where}\quad\tau:=\inf\{ t\geq 0 : |S_t|\geq m\}\wedge T
  \]
  is a martingale. Fix $m\in\mathbb{N}$, $0\leq s< t\leq T$, and define the stopping times
  \begin{align*}
    \sigma&:=\inf\{ r\geq  s : |S_r|\geq m \}\wedge T,\\
    \sigma_\varepsilon&:=\inf\{ r\geq s :  S_r> m-\varepsilon \text{ or } S_r \leq \varepsilon-m\}\wedge T
  \end{align*}
  for $0<\varepsilon\leq1$. First note that, by continuity of $S$ and right-continuity of $(\mathcal{F}_t)$, one has that
  $\sigma_\varepsilon$, $\sigma$, and $\tau$ are in fact stopping times. By Lemma~\ref{lem:S.stopped.lsc} the function $\omega\mapsto S_{t\wedge \sigma_\varepsilon(\omega)}(\omega)$ is lower semicontinuous w.r.t.~$\|\cdot\|_\infty$ for every $\varepsilon$.
  In particular, for every continuous $\mathcal{F}^0_s$-measurable function $h\colon\Omega\to[0,1]$, it holds that
  \[
    (H\cdot S)_T\text{ is lower semicontinuous, where }
	H:=h \1_{(s,\sigma_\varepsilon\wedge t]}\in\mathcal{H}^f. 
  \] 
  Since additionally $|S_t^{\sigma_\varepsilon}-S_s|\leq 2m$, there exists a sequence of continuous functions $X_n\colon\Omega\to[-2m,2m]$ such that $X_n\leq (H\cdot S)_T$ which increases pointwise to $(H\cdot S)_T$. Since $X_n\in G$ for all $n$, it follows that
  \[
    E_Q[ h(S_t^{\sigma_\varepsilon}-S_s)] =E_Q[ (H\cdot S)_T] =\sup_n E_Q[X_n] \leq 0.
  \]
  By assumption (A1), for every bounded and $\mathcal{F}_s^0$-measurable function $h$, there exists a sequence of  continuous $\mathcal{F}_s^0$-measurable functions $h_n\colon\Omega\to[0,1]$ which converges $Q$-almost surely to $h$, in particular 
  \[
    E_Q[ h(S_t^{\sigma_\varepsilon}-S_s)]
    =\lim_n E_Q[ h_n(S_t^{\sigma_\varepsilon}-S_s)]
    \leq 0 .
  \]
  The fact that $\sigma_\varepsilon$ increases to $\sigma$ as $\varepsilon$ tends to $0$ (and therefore  $S_t^{\sigma_\varepsilon}\to S_t^\sigma$ by continuity of $S$), shows that
  \[
    E_Q[ h(S_t^\sigma-S_s)] =\lim_{\varepsilon\to0}  E_Q[ h(S_t^{\sigma_\varepsilon}-S_s)]\leq 0.
  \]
  Furthermore, notice that $\sigma=\tau$ on $\{\tau\geq s\}$, so that $\1_{\{\tau\geq s\}}(S_t^\sigma-S_s) =S^\tau_t-S_s^\tau$. Since $\tau$ is the hitting time of a closed set, it is also a stopping time w.r.t.~the raw filtration $(\mathcal{F}_t^0)$, so that $h\1_{\{\tau\geq s\}}\colon\Omega\to[0,1]$ is $\mathcal{F}_s^0$-measurable. This shows that
  \[
    E_Q[h(S^\tau_t-S_s^\tau)]
    =E_Q[(h \1_{\{\tau\geq s\}})(S_t^\sigma-S_s)]
    \leq 0, 
  \]
  which implies $E_Q[S^\tau_t|\mathcal{F}_s^0]\leq S^\tau_s$, i.e.~$S^\tau$ is a supermartingale w.r.t.~the raw filtration
  $(\mathcal{F}^0_t)$.
  Finally, using that $S^\tau$ is bounded and ${\cal F}_s\subseteq {\cal F}^0_{s+\varepsilon}$ yields
  \[
    E_Q[S^\tau_t-S^\tau_s|\mathcal{F}_s]
    =\lim_{\varepsilon\to0} E_Q[S^\tau_t-S^\tau_{s+\varepsilon}|\mathcal{F}_s]
    =\lim_{\varepsilon\to0} E_Q\big[E_Q[S^\tau_t-S^\tau_{s+\varepsilon}|\mathcal{F}^0_{s+\varepsilon}]|\mathcal{F}_s\big]
    \leq 0
  \]
  which shows that $S^\tau$ is a supermartingale.

  By similar arguments one can also show that $S^\tau$ is a submartingale (and thus a martingale). Indeed, replace $h$ by a continuous $\mathcal{F}^0_s$-measurable function $\tilde h\colon\Omega\to[-1,0]$, and the stopping times $\sigma_\varepsilon$ by the stopping times $\tilde\sigma_\varepsilon:=\inf\{ r\geq s :  S_r\geq m-\varepsilon \text{ or } S_r < \varepsilon -m\}\wedge T$ for $\varepsilon>0$. The same arguments as in Lemma~\ref{lem:S.stopped.lsc} show that $\omega\mapsto S_{t\wedge\tilde\sigma_\varepsilon(\omega)}(\omega)$ is upper semicontinuous, which implies that $(H\cdot S)_T$ is lower semicontinuous for $H:=\tilde h\1_{(s,\tilde\sigma_\varepsilon\wedge t]}\in\mathcal{H}^f$. The rest follows the same way as before.
\end{proof}

\begin{lemma}\label{lem:compacts.stopped}
  Assume that (A1) and (A2) hold true. Then there exists an increasing sequence of non-empty compacts $(K_n)$ such that $\Omega=\bigcup_n K_n$, and $\omega^t\in K_n$ for every $(t,\omega)\in[0,T]\times K_n$.
\end{lemma}

\begin{proof}
  By assumption $\Omega=\bigcup_n K_n'$ for some non-empty compacts $(K_n')$, where we assume without loss of generality that $K_n'\subset K_{n+1}'$ for every $n$. Define the function $\rho\colon [0,T]\times\Omega\to\Omega$, $(t,\omega)\mapsto \omega^t$ which, again by assumption, is continuous. Therefore $K_n:=\{\omega^t : t\in[0,T],\,\omega\in K_n'\}=\rho([0,T],K_n')$ has the desired properties.
\end{proof}

\begin{lemma}\label{lem:stop}
  Assume that (A1) and (A2) hold true and fix a sequence of compacts $(K_j)$ as in Lemma~\ref{lem:compacts.stopped}.
  Further fix a continuous function $Z\colon\Omega\to \mathbb{R}$, $H\in\mathcal{H}^f$, and $n\in\mathbb{N}$. 
  If $(H\cdot S)_T(\omega)\geq -Z(\omega)$ for all $\omega\in K_j$, 
  then $(H\cdot S)_t(\omega)\geq -Z(\omega^t)$ for all $(t,\omega)\in [0,T]\times K_j$.
\end{lemma}

\begin{proof}
  Fix $H=\sum_{n=1}^{N} h_n \1_{(\tau_n,\tau_{n+1}]}\in\mathcal{H}^f$,
  $\omega\in K_j$, and $t\in[0,T)$ (for $t=T$ the statement holds by assumption).
  We may assume that $\tau_{N+1}=T$ by adding an additional stopping time and setting $h_N\equiv 0$.
  Further, fix $\varepsilon>0$ with $t+\varepsilon\leq T$, and $m\in\mathbb{N}$ such that
  $\tau_m(\omega^{t+\varepsilon})\leq t\leq \tau_{m+1}(\omega^{t+\varepsilon})$.
  Then 
  \begin{align*} 
    (H\cdot S)_t(\omega^{t+\varepsilon})-(H\cdot S)_T(\omega^{t+\varepsilon})
    =&h_m(\omega^{t+\varepsilon})(S_{t}(\omega^{t+\varepsilon})-S_{\tau_{m+1}(\omega^{t+\varepsilon})}(\omega^{t+\varepsilon}))\\
    &-\sum_{n=m+1}^{N} h_m(\omega^{t+\varepsilon})(S_{\tau_{n+1}(\omega^{t+\varepsilon})}(\omega^{t+\varepsilon})-S_{\tau_n(\omega^{t+\varepsilon})}(\omega^{t+\varepsilon}))
  \end{align*}
  and 
  \begin{equation*}
    |S_t(\omega^{t+\varepsilon})-S_{\tau_{m+1}(\omega^{t+\varepsilon})}(\omega^{t+\varepsilon})|
    \leq\delta(\varepsilon)
    \quad\text{and}\quad
    |S_{\tau_{n+1}(\omega^{t+\varepsilon})}(\omega^{t+\varepsilon})-S_{\tau_n(\omega^{t+\varepsilon})}(\omega^{t+\varepsilon})|\leq\delta(\varepsilon)
  \end{equation*}
  for all $n\ge m+1$, where $\delta(\varepsilon):=\max_{r,s\in[t,t+\varepsilon]} |\omega(r)-\omega(s)|$. 
  Let $C$ be a constant such that $|h_n|\leq C$. Then, since 
  $\lim_{\varepsilon\downarrow 0}\delta(\varepsilon)=0$, it holds 
  \[
    |(H\cdot S)_t(\omega^{t+\varepsilon})-(H\cdot S)_T(\omega^{t+\varepsilon})| 
	\leq NC \delta(\varepsilon)\to 0
  \]
  as $\varepsilon\downarrow 0$. Since $\mathcal{F}_t\subset\mathcal{F}_{t+\varepsilon}^0$, 
  it follows that $(H\cdot S)_t(\omega)= (H\cdot S)_t(\omega^{t+\varepsilon})$ for all $\varepsilon>0$, so that
  \[
    (H\cdot S)_t(\omega) 
    =\lim_{\varepsilon\downarrow 0} (H\cdot S)_T(\omega^{t+\varepsilon}) 
    \geq \liminf_{\varepsilon\downarrow 0} -Z(\omega^{t+\varepsilon})
    =-Z(\omega^t)
  \] 
  since  $\omega^{t+\varepsilon}\in K_j$ for all $\varepsilon>0$, and $\varepsilon\mapsto Z(\omega^{t+\varepsilon})$ is continuous by assumption.
\end{proof}

We have now all ingredients at hand to prove the main results of the present paper. 

\begin{proof}[Proof of Theorem~\ref{thm:main.integrals}]
  Fix a continuous function $Z\colon\Omega\to[0,+\infty)$, a sequence of compact sets $(K_n)$ as in Lemma~\ref{lem:compacts.stopped}.
  \smallskip

  Step (a): Fix $n\in\mathbb{N}$ and define
  \[
    \phi_n(X):=\inf\left\{ \lambda\in\mathbb{R} \,:\, 
    \begin{array}{l}
      \text{there is $H\in\mathcal{H}^f$ and $c\in\mathbb{R}$ such that}\\
      \text{$(H\cdot S)_T\geq c$ on $\Omega$ and $\lambda+(H\cdot S)_T\geq X$ on $K_n$}
    \end{array} \right\}
  \]
  for $X\colon\Omega\to\mathbb{R}$. By Lemma~\ref{lem:weak} it follows that
  \begin{align}\label{eq:weak.psik}
    \phi_n(X)\geq\sup_{Q\in\mathcal{M}(K_n)} E_Q[X]		
  \end{align}
  for every Borel measurable $X$ which is bounded from below on $K_n$. Let $\bar\omega\in K_n$ be the constant path $t\mapsto\bar\omega(t):=\omega(0)$ for some $\omega\in K_n$. Since the Dirac measure $\delta_{\bar\omega}$ assigning probability $1$ to $\bar\omega$ belongs to $\mathcal{M}(K_n)$, it follows that $\phi_n$ is real-valued on $C_b$ and $\phi_n(m)=m$ for every $m\in\mathbb{R}$.
	
  Further, it is straightforward to check that $\phi_n$ is convex and increasing in the sense that $\phi_n(X)\leq \phi_n(Y)$ whenever $X \leq Y$. Moreover, $\phi_n$ is continuous from above on $C_b$, i.e.~$\phi_n(X_k)\downarrow \phi_n(0)$ for every sequence $(X_k)$ in $C_b$ such that $X_k\downarrow 0$. To that end, fix such a sequence $(X_k)$ and let $\varepsilon>0$ be arbitrary. By Dini's lemma one has $X_k\leq\varepsilon$ on $K_n$ for all $k$ large enough, so that $\phi_n(X_k)\leq\varepsilon$ for all such $k$, which shows that $\phi_n(X_k)\downarrow 0$. It follows from \cite[Proposition~1.1]{cheridito2015representation} (see also \cite[Theorem~2.2]{Bartl2019}) that 
  \begin{equation}\label{rep:Cb}
    \phi_n(X)=\max_{Q\in ca^+(\Omega)}\left(E_Q[X]-\phi^\ast_{n}(Q)\right)
  \end{equation}
  for all $X\in C_b$, where $\phi^\ast_n(Q):=\sup_{X\in C_b} (E_Q[X]-\phi_n(X))$ and $ca^+(\Omega)$ denotes the set of non-negative countably additive Borel measures on $\Omega$. We claim that
  \begin{align}\label{eq:conjugate}
    \phi_n^{\ast}(Q)=\begin{cases}
    0,&\text{if } Q\in\mathcal{M}(K_n),\\
    +\infty, &\text{else},
    \end{cases}
  \end{align}
  for all $Q\in ca^+(\Omega)$.
  First notice that~\eqref{eq:weak.psik} implies $\phi_n^{\ast}(Q)\leq 0$ whenever $Q\in\mathcal{M}(K_n)$. Since in addition $\phi_n(0)=0$, it follows that $\phi_n^{\ast}(Q)=0$. On the other hand, if $Q\notin\mathcal{M}(K_n)$, then $\phi_n^{\ast}(Q)=+\infty$. Indeed, if $Q$ is not a probability, then $\phi_n(m)=m$ implies that $\phi_n^\ast(Q)\geq\sup_{m\in\mathbb{R}} (mQ(\Omega)-m)=+\infty$.
  Similarly, since $K_n^c$ is open, there exists a sequence of bounded continuous functions 
  $(X_k)$ such that $X_k\uparrow +\infty \1_{K_n^c}$ with the convention $0\cdot (+\infty):=0$. 
  By definition $\phi_n(X_k)\leq 0$ for all $k$, from which it follows that
  \begin{align*}
    \phi_n^\ast(Q) \geq \sup_k E_Q[X_k] = +\infty E_Q[\1_{K_n^c}].
  \end{align*}
  It remains to show that if $Q$ is a probability with $Q(K_n)=1$ but not a martingale measure, then $\phi_n^\ast(Q)=+\infty$.
  Note that compactness of $K_n$ implies boundedness of $K_n$ w.r.t.~$\|\cdot\|_\infty$, and therefore
  $Q$ is also not a local martingale measure. Thus Proposition~\ref{lem:polar} yields the existence of $X\in C_b$ and 
  $H\in\mathcal{H}^f$ such that $X\leq (H\cdot S)_T$ and $E_Q[X]>0$. Since $\phi_n(mX) \leq 0$ for all $m>0$, it follows that
  $\phi^\ast_n(Q) \geq \sup_{m>0} (E_Q[mX]-\phi_n(mX))=+\infty$.

  Next, fix some upper semicontinuous $X$ which is bounded from above (i.e.~$X=X\wedge m$ for some $m>0$)
  and satisfies $X\geq -Z$. We claim that 
  \begin{equation}\label{rep:UZ}
    \phi_n(X)=\max_{Q\in\mathcal{M}(K_n)} E_Q[X].
  \end{equation}
  To that end, let $(X_k)$ be a sequence in $C_b$ such that $X_k\downarrow X$. By \eqref{rep:Cb} and \eqref{eq:conjugate} there exist
  $Q_k\in\mathcal{M}(K_n)$ such that $\phi_n(X_k)=E_{Q_k}[X_k]$. Since $\mathcal{M}(K_n)$ is (sequentially) compact in the weak topology induced by the continuous bounded functions, possibly after passing to a subsequence, we may assume that $Q_k\to Q$  for some $Q\in \mathcal{M}(K_n)$. For every $\varepsilon>0$ there exists $k^\prime$ such that $E_Q[X_{k^\prime}]\le E_Q[X]+\varepsilon$. Choose $k\ge k^\prime$ such that
  $E_{Q_k}[X_{k^\prime}]\le E_Q[X_{k^\prime}]+\varepsilon$. Then 
  \[
    E_{Q_k}[X_k]\le E_{Q_k}[X_{k^\prime}]\le E_Q[X_{k^\prime}]+\varepsilon\le E_Q[X]+2\varepsilon
  \]
  so that
  \[
    \phi_n(X)
    \le\lim_k \phi_n(X_k)
    =\lim_k E_{Q_k}[X_k]
    \le E_Q[X] +2\varepsilon
    \le \sup_{R\in\mathcal{M}(K_n)} E_R[X] +2\varepsilon
    \le\phi_n(X)+2\varepsilon,
  \]
  where the last inequality follows from \eqref{eq:weak.psik}. This shows \eqref{rep:UZ}.\smallskip

  Step (b): For $X\colon\Omega\to(-\infty,+\infty]$ define
  \[ 
    \phi(X):=\inf\left\{ \lambda\in\mathbb{R} \,:\, 
    \begin{array}{l}
    \text{there is $(H^n)$ in $\mathcal{H}^f$ such that $\lambda+(H^n\cdot S)_t(\omega)\geq -Z(\omega^t)$}\\
    \text{for all $(t,\omega)\in[0,T]\times\Omega$ and }\lambda+ \liminf_n (H^n\cdot S)_T \geq X\text{ on }\Omega
    \end{array} \right\}.
  \]
  Let $X\in C_{\delta\sigma}$ such that $X\geq -Z$ for some $c\ge 0$, and let $(Y_n)$ be a sequence of upper semicontinuous functions which increases pointwise to $X$. Define $X_n:= (Y_n\wedge n)\vee (-cZ)$ which is still upper semicontinuous and increases to $X$. We claim that $\sup_n \phi_n(X_n)=\phi(X)$. First observe that for every $Q\in\mathcal{M}_c(\Omega)$ Fatou's lemma and Lemma~\ref{lem:weak} imply
  \[
    \lambda=\lambda+\liminf_n E_Q[(H^n\cdot S)_T]\ge E_Q[\lambda+\liminf_n (H^n\cdot S)_T]\ge E_Q[X]
  \]
  for every $\lambda\in\mathbb{R}$ and $(H^n)$ in $\mathcal{H}^f$ such that $\lambda+\liminf_n (H^n\cdot S)_T\ge X$ and $\lambda +(H^n\cdot S)_T\geq -mZ$ for all $n$ and some $m\ge 0$. Hence, one gets
  \begin{align}\label{eq:phi.supkphik} 
    \phi(X)
    \geq \sup_{Q\in\mathcal{M}_c(\Omega)} E_Q[X] 
    \geq \sup_n \sup_{Q\in\mathcal{M}(K_n)} E_Q[X_n]
    =\sup_n \phi_n(X_n),
  \end{align}
  where the last equality follows from~\eqref{rep:UZ}.

  On the other hand, let $m>\sup_n\phi_n(X_n)$ so that, by definition, for each $n$ there exists $H^n\in\mathcal{H}^f$ such that $m+(H^n\cdot S)_T\geq X_n\geq -Z$ on $K_n$. Thus, it follows from Lemma~\ref{lem:stop} that
  \begin{equation}\label{dynbound}
    m+(H^n\cdot S)_t(\omega)\geq -Z(\omega^t) \quad\text{for all }(t,\omega)\in[0,T]\times K_n.
  \end{equation}
  Fix $\varepsilon>0$. Define the stopping times
  \[
    \sigma_n(\omega):=\inf\{ t\in[0,T] : m+\varepsilon+(H^n\cdot S)_t(\omega) +Z(\omega^t)=0\}\wedge T 
  \]
  and notice that  
  \begin{equation}\label{dynbound2} 
    (H^n\cdot S)_{\sigma_n}=(\tilde{H}^n\cdot S)_T
    \quad\text{for }  \tilde{H}^n:=\sum_{i=1}^{N} h_i^n \1_{\{\sigma_n\geq \tau_i\}}
    \1_{(\tau_i\wedge \sigma_n,\tau_{i+1}\wedge \sigma_n]}\in\mathcal{H}^f,
  \end{equation}
  where $H^n=\sum_{i=1}^{N} h_i^n \1_{(\tau_i,\tau_{i+1}]}$. Fix $\omega\in\Omega$. Then $\omega\in K_j$ for some $j\in\mathbb{N}$ and therefore, by \eqref{dynbound} it follows that $\sigma_n(\omega)=T$ whenever $n\geq j$. Hence, we have 
  \[
    m+\varepsilon+(\tilde{H}^n\cdot S)_T(\omega)
    = m+\varepsilon+ (H^n\cdot S)_T(\omega) 
    \geq X_n(\omega) \quad\text{for }n\geq j. 
  \]	
  As $\omega$ was arbitrary, it follows that $\liminf_n ( m+\varepsilon+(\tilde{H}^n\cdot S)_T)\geq X$. Moreover, it follows from~\eqref{dynbound2} that 
  \[
    m+\varepsilon + (\tilde{H}^n\cdot S)_t(\omega)
    \geq - Z(\omega^{t\wedge \sigma_n(\omega)})
    \geq - Z(\omega^t)\quad\text{for all }(t,\omega)\in[0,T]\times\Omega,
  \]
  which shows that $\phi(X)\leq m+\varepsilon$. Finally, since $m>\sup_n\phi_n(X_n)$ and $\varepsilon>0$ was arbitrary, we conclude that $\phi(X)\leq\sup_n\phi_n(X_n)$, which shows that all inequalities in~\eqref{eq:phi.supkphik} are equalities. In particular, $\phi(X)=\sup_{Q\in\mathcal{M}_c(\Omega)} E_Q[X]$, which shows~\eqref{eq:duality}.
  \smallskip

  Step (c): We finally show that $\mathcal{M}_c(\Omega)$ can be replaced by the set $\mathcal{M}_Z(\Omega)$, and $\mathcal{H}^f$ by $\mathcal{H}$. To that end, fix $X\colon\Omega\to(-\infty,+\infty]$ satisfying $X\geq -Z$
  for some $\lambda\in\mathbb{R}$, $Q\in\mathcal{M}_Z(\Omega)$,  and $(H^n)$ in $\mathcal{H}$ such that 
  $\lambda+(H^n\cdot S)_t(\omega)\geq -Z(\omega^t)$ for all $(t,\omega)\in[0,T]\times\Omega$ and $\lambda+ \liminf_n (H^n\cdot S)_T \geq X$. Define 
  \begin{equation*}
    H^{n,K}:=\sum_{k=1}^K h^n_k \1_{(\tau^n_k,\tau^n_{k+1}]}\in\mathcal{H}^f\quad \text{and}\quad H^n=\sum_{k=1}^{\infty} h^n_k \1_{(\tau^n_k,\tau^n_{k+1}]}.
  \end{equation*}
  Therefore, one gets
  \[
    \lambda+(H^{n,K}\cdot S)_T(\omega)
    =\lambda+(H^{n}\cdot S)_{\tau_{K+1}^n(\omega)}(\omega)
    \geq -Z(\omega^{\tau_{K+1}^n(\omega)})
    \geq -Z(\omega),
  \]
  where the last inequality holds by assumption.
  Hence, by Lemma~\ref{lem:weak} and Fatou's lemma, it follows that
  \begin{align*}
    \lambda&=\lambda+\liminf_n\liminf_K E_Q[(H^{n,K}\cdot S)_T]\ge \liminf_n E_Q[\lambda+ \liminf_K (H^{n,K}\cdot S)_T]\\
    &=\liminf_n E_Q[\lambda+ (H^{n}\cdot S)_T]\ge E_Q[\lambda+ \liminf_n(H^{n}\cdot S)_T]\ge E_Q[X].  
  \end{align*}
  This shows 
  \begin{align*}
    &\inf\left\{ \lambda\in\mathbb{R} \,:\, 
    \begin{array}{l}
    \text{there is a sequence $(H^n)$ in $\mathcal{H}^f$ such that}\\ \text{$\lambda+(H^n\cdot S)_t(\omega)\geq -Z(\omega^t) \text{ for all } (t,\omega)\in[0,T]\times\Omega$ and }\\ \lambda+ \liminf_n (H^n\cdot S)_T(\omega) \geq X(\omega) \,\text{ for all } \omega\in\Omega
    \end{array}
    \right\}\\
    &\qquad\ge  
    \inf\left\{ \lambda\in\mathbb{R} \,:\, 
    \begin{array}{l}
    \text{there is a sequence $(H^n)$ in $\mathcal{H}$ such that}\\ \text{$\lambda+(H^n\cdot S)_t(\omega)\geq -Z(\omega^t) \text{ for all } (t,\omega)\in[0,T]\times\Omega$ and }\\ \lambda+ \liminf_n (H^n\cdot S)_T(\omega) \geq X(\omega) \,\text{ for all } \omega\in\Omega
    \end{array}
    \right\}\\
    &\qquad \ge \sup_{Q\in\mathcal{M}_Z(\Omega)} E_Q[X]\ge \sup_{Q\in\mathcal{M}_c(\Omega)} E_Q[X],
  \end{align*}
  where the first and last terms coincide by the previous steps (a) and (b). 
\end{proof}

The proof of Corollary~\ref{cor:omega.is.whole.space} is a consequence of the following lemma.

\begin{lemma}\label{lem:compact.supported.mm.are.dense}
  Let $\Omega=C([0,T],\mathbb{R}^d)$, $Q\in\mathcal{M}(\Omega)$, and $X\colon\Omega\to\mathbb{R}$ be bounded and Borel.
  For every $\varepsilon>0$ there exists $K\subset\Omega$ compact and $\tilde{Q}\in\mathcal{M}(K)$ such that $|E_Q[X]-E_{\tilde{Q}}[X]|\leq\varepsilon$. In particular, $\sup_{Q\in\mathcal{M}(\Omega)} E_Q[X]=\sup_{Q\in\mathcal{M}_c(\Omega)} E_Q[X]$.
\end{lemma}

\begin{proof}
  If $X = 0$, there is nothing to prove. Otherwise, since $\Omega$ is a Polish space, there exists  $K\subset\Omega$ compact such that $Q(K^c)\leq \varepsilon/ \|X\|_\infty$. By an Arzel{\`a}-Ascoli type theorem there exist $a\in\mathbb{R}$ and a continuous increasing function $f\colon[0,+\infty)\to[0,+\infty)$ such that 
  \[
    K\subset \tilde{K}:=\big\{\omega\in\Omega: \|\omega\|_\infty \leq a \text{ and } |\omega(t)-\omega(s)|\leq f(|t-s|)\text{ for all }s,t\in[0,T] \big\} 
  \]
  and $\tilde{K}$ is compact. Now define the stopping time
  \[
    \tau:=\inf\{ t\geq 0 : |S_t|>a \text{ or } |S_t-S_s|> f(|t-s|)\text{ for some }s\in\mathbb{Q}\cap[0,t] \}\wedge T
  \]
  so that $\tilde{K}=\{\tau=T\}$. Then, for $\tilde{Q}:=Q\circ (S^\tau)^{-1}\in\mathcal{M}(\tilde{K})$ one has 
  \[
    |E_{\tilde{Q}}[X]-E_Q[X]|\le |E_Q[X(S^\tau)1_{K^c}]| + E_Q[X1_{K^c}]\leq 2\varepsilon.
  \]
  In particular, $\sup_{Q\in\mathcal{M}(\Omega)} E_Q[X]=\sup_{Q\in\mathcal{M}_c(\Omega)} E_Q[X]$. 
\end{proof}

\begin{proof}[Proof of Corollary~\ref{cor:omega.is.whole.space}]
  Denote by $\mathcal{K}$ the set of all compact subsets $K\subset\Omega$.
  For $K\in\mathcal{K}$ define $\tilde{K}:=\{\omega^t : t\in[0,T]\text{ and }\omega\in K\}$ which is compact due to (the proof of) Lemma~\ref{lem:compacts.stopped}. For $K\in\mathcal{K}$ and every bounded upper semicontinuous  function $X\colon \Omega\to\mathbb{R}$ define
  \[
    \phi_K(X):=\inf\left\{\lambda \in\mathbb{R} : 
	\begin{array}{l}
	\text{there is }H\in\mathcal{H}^f\text{ and } c\geq 0\text{ such that}\\
	\lambda + (H\cdot S)_T(\omega)\geq -c \text{ for all }\omega\in \Omega\text{ and}\\
	\lambda + (H\cdot S)_T(\omega)\geq X(\omega)\text{ for all }\omega\in K 
	\end{array}\right\}.
  \]
  Then, one has
  \begin{align*}
    \sup_{K\in\mathcal{K}}\phi_K(X)=\sup_{K\in\mathcal{K}}\phi_{\tilde{K}}(X)=\sup_{K\in\mathcal{K}} \max_{Q\in\mathcal{M}(\tilde{K})} E_Q[X]=\sup_{K\in\mathcal{K}} \max_{Q\in\mathcal{M}(K)} E_Q[X]	= \sup_{Q \in {\cal M}_c(\Omega)}E_Q[X].
  \end{align*}	
  The first and third equalities follow from $K\subseteq \tilde K$, the second one follows by  $\phi_{\tilde{K}}(X)=\max_{Q\in\mathcal{M}(\tilde{K})} E_Q[X]$ as in \eqref{rep:UZ} for every $K\in\mathcal{K}$ and the last equality follows by definition of ${\cal M}_c(\Omega)$. Now, use Lemma~\ref{lem:compact.supported.mm.are.dense} to conclude.
\end{proof}

\begin{proof}[Proof of Theorem~\ref{thm:main.integrals.Z}]
  Step (a): For $n\in\mathbb{N}$ and every function $X\colon\Omega\to\mathbb{R}$ define
  \[ 
    \phi_n(X):=\inf\left\{ \lambda\in\mathbb{R} \,:\, 
    \begin{array}{l}
    \text{there is $H\in\mathcal{H}^f$ and $c>0$ such that}\\
    \text{$(H\cdot S)_T\geq -c$ and $\lambda+(H\cdot S)_T\geq X -Z/n$}
    \end{array} \right\}.
  \]
  It follows from Lemma~\ref{lem:weak} that
  $\phi_n(X)\geq \sup_{Q\in\mathcal{M}_Z(\Omega)} \big( E_Q[X]-E_Q[Z]/n\big)$
  for every Borel function $X$ which is bounded from below. Moreover, if $(X_k)$ is a sequence
  in $C_b$ decreasing pointwise to $0$, then $\phi(X_n)\downarrow \phi(0)$. Indeed,
  fix $\varepsilon>0$ arbitrary and $H\in\mathcal{H}^f$ with $(H\cdot S)_T\geq -c$ for some $c\geq 0$ such that
  \[ 
    \varepsilon + \phi_n(0)+ (H\cdot S)_T+ Z/n\geq 0. 
  \]
  Now define $\tilde{c}:=\|X_1\|_\infty-\varepsilon - \phi_n(0)+c$ so that
  $\tilde{c}+\varepsilon +\phi_n(0)+ (H\cdot S)_T\geq X_1$.
  Since $\{Z\leq \tilde{c}n\}$ is compact, it follows from Dini's lemma that $X_k \1_{\{Z\leq \tilde{c}n\}}\leq\varepsilon$
  for $k$ large enough. Hence
  \begin{align*}
    X_k \leq X_k \1_{\{Z\leq \tilde{c}n\}}+X_1 \1_{\{Z>\tilde{c}n\}}
    &\leq \varepsilon+(\varepsilon+\phi_n(0)+(H\cdot S)_T+ Z/n)\1_{\{Z>\tilde{c}n\}}\\
    &\leq 2\varepsilon+\phi_n(0)+ (H\cdot S)_T+ Z/n
  \end{align*}
  so that $\phi_n(X_k)\leq \phi_n(0)+2\varepsilon$ for $k$ large enough which shows that $\phi_n(X_k)\downarrow \phi_n(0)$. 
  Now, a computation similar to the one in the proof of Theorem~\ref{thm:main.integrals} shows that
  \begin{align}
  \label{eq:rep.phi.n.Z}
    \phi_n(X)=\max_{Q\in\mathcal{M}_Z(\Omega)}\big (E_Q[X]-E_Q[Z]/n\big)
  \end{align}
  for every bounded upper semicontinuous function $X\colon\Omega\to\mathbb{R}$.
  Indeed, first notice that since by assumption $Z\geq \|\cdot\|_\infty$, the set
  $\mathcal{M}_Z(\Omega)$ coincides with the set of all local martingale
  measures which integrate $Z$. Therefore, the same arguments as in the proof of Theorem~\ref{thm:main.integrals}
  show that
  \[
    \phi_n^\ast(Q):=\sup_{X\in C_b} (E_Q[X]-\phi_n(X))
    =\begin{cases}
    E_Q[Z]/n, &\text{if } Q\in\mathcal{M}_Z(\Omega),\\
    +\infty,&\text{else},
    \end{cases}
  \]
  and thus that~\eqref{eq:rep.phi.n.Z} is true, at least whenever $X\in C_b$. As for the extension to upper semicontinuous functions, notice that
  $\phi(X)=\max_{Q\in\Lambda_{2c}} (E_Q[X]-E_Q[Z]/n)$ for every $X\in C_b$ satisfying $|X|\leq c$ where $\Lambda_{2c}:=\{ \phi_n^\ast\leq 2c\}$.
  Using the fact that $Z$ has compact sublevel sets and Proposition~\ref{lem:polar}, it follows that $\Lambda_c$ is (sequentially) compact. The rest follows analogously to the proof of Theorem~\ref{thm:main.integrals}.  \smallskip
  
  Step (b): For $X\in C_{\delta\sigma}$ define 
  \[
    \phi(X):=\inf\left\{ \lambda\in\mathbb{R} \,:\, 
    \begin{array}{l}
      \text{there is $(H^n)$ in $\mathcal{H}^f$ and $c\geq 0$ such that $(H^n\cdot S)_T\geq -cZ$}\\
      \text{for all $n$ and }\lambda+ \liminf_n (H^n\cdot S)_T \geq X
    \end{array} \right\}.
  \] 
  Fix $X\in C_{\delta\sigma}$ bounded from below and $X_n$ upper semicontinuous bounded from below such that
  $X=\sup_n X_n$. Then, it follows from Fatou's lemma and Lemma~\ref{lem:weak} that
  \begin{align*}
    \phi(X)
    &\geq \sup_{Q\in\mathcal{M}_Z(\Omega)} E_Q[X]
    =\sup_{Q\in\mathcal{M}_Z(\Omega)}\Big(\sup_n  E_Q[X_n] -E_Q[Z]/n\Big)\\
    &=\sup_n\sup_{Q\in\mathcal{M}_Z(\Omega)}\big(  E_Q[X_n] -E_Q[Z]/n\big)
    =\sup_n \phi_n(X_n).
  \end{align*}
  On the other hand, if $m>\sup_n\phi_n(X_n)$, then for every $n$ there exists $H^n\in\mathcal{H}^f$ such that $m+(H^n\cdot S)_T\geq X_n-Z/n$. Hence, $(H^n\cdot S)_T\geq -cZ$ for $c:=\|X_1\wedge 0\|_\infty+m+1$ and 
  $m+ \liminf_n(H^n\cdot S)_T\geq \liminf_n (X_n-Z/n)=X$, which completes the proof.
\end{proof}

\appendix
\section{Appendix}\label{sec:appendix}

\subsection{Kolmogorov continuity criterion}

In this section we briefly recall a version of the so-called Kolmogorov continuity criterion, which provides a sufficient condition for H\"older and Sobolev regularity of stochastic processes. The presented version is a slight reformulation of~\cite[Theorem~A.10]{Friz2010}. 

Let $(\tilde \Omega,\mathcal{F}, P)$ be a probability space, $X\colon [0,T]\times \tilde \Omega \to \R^d$ be a stochastic process, $T\in (0,+\infty)$, $(\R^d,|\cdot|)$ be the Euclidean space and $W$ be a $d$-dimensional Brownian motion.

\begin{theorem}\label{thm:Kolmogorov} 
  Let $q>r\geq1$ and suppose that there exists a constant $M>0$ such that
  \begin{equation*}
    E_P \big[ |X_t-X_s|^q \big]\leq M |t-s|^{\frac{q}{r}}\quad \text{for all}\quad s,t\in[0,T].
  \end{equation*}
  Then, for any $\alpha\in [0,1/r-1/q)$ and $\delta:=\alpha+1/q$ there exists a constant $C=C(r,q,\alpha,T)$ such that 
  \begin{equation*}
    E_P \big[ \|X\|_{\alpha}^q \big]\leq CM  \quad\text{and}\quad E_P \big[ \|X\|_{W^{\delta,q}}^q \big]\leq CM,
  \end{equation*}
  where we recall the semi-norms
  \begin{equation}\label{eq:Sobolev norm}
    \|X\|_{\alpha} := \sup_{s,t\in [0,T]}\frac{|X_t-X_s|}{|t-s|^\alpha}\quad \text{and}\quad
    \|X\|_{W^{\delta,q}} := \bigg(\int_{[0,T]^2} \frac{|X_t-X_s|^q}{|t-s|^{\delta q+1}}\dd s \dd t\bigg)^{\frac{1}{q}}.
  \end{equation}
\end{theorem}

Applying Theorem~\ref{thm:Kolmogorov} to It\^o processes reveals the following regularity criterion.

\begin{corollary}\label{cor:ito processes}
  Let $X$ be a $d$-dimensional It\^o process of the form 
  \begin{equation*}
    X_t = x_0 + \int_0^t a_s \dd W_s , \quad t\in [0,T],
  \end{equation*}
  for a predicable process $a\colon [0,T]\times \tilde\Omega \to \R^{d\times d}$ and $x_0 \in \R^d$. Suppose $q\in (2,+\infty)$, $\alpha \in (0,1/2-1/(2q))$ and $\delta= \alpha-1/q$. If $ E_P\big[\int_0^T|a_s|^{q}\dd s\big]<+\infty$,
  then
  \begin{equation*}
    X\in C^{\alpha}([0,T],\R^d) \quad \text{and} \quad X\in W^{\delta,q}([0,T],\R^d),\quad P\text{-a.s.}
  \end{equation*}
\end{corollary}

\begin{proof}
  Using the Burkholder-Davis-Gundy inequality and Jensen's inequality, one has
  \begin{equation*}
    E_P \big[ |X_t-X_s|^q \big]
    \leq E_P \bigg[\bigg(\int_s^t|a_r|^{2}\dd r\bigg)^{q/2}\bigg] 
    \leq E_P \bigg[\int_0^T |a_r|^{q}\dd r \bigg] |t-s|^{q(\frac{1}{2}-\frac{1}{q})}.
  \end{equation*}
  Therefore, Theorem~\ref{thm:Kolmogorov} implies the assertion.
\end{proof}

\subsection{Construction of counter-example}

The example (see Remark~\ref{rem:counterexample}) showing that bounded variation strategies and, in particular, simple trading strategies are not rich enough to obtain the pathwise pricing- hedging duality was based on a H\"older continuous function with exploding quadratic variation. The existence of such a function is ensured by the following lemma. 

\begin{lemma}\label{lem:example function}
  There exists a function $\tilde \omega \in C^{1/4}([0,T],\mathbb{R})$ for some $T>0$ and a refining sequence of partitions $(\tilde\pi_n)_{n\in \mathbb{N}}$ of the interval~$[0,T]$ such that 
  \begin{align*}
    & 0\leq \tilde \omega (t) \leq 1, \quad t\in [0,T],\\
    &\langle\tilde\omega\rangle_t:=\lim_n\langle\tilde\omega\rangle_t^n,
    \quad \text{where}\quad
    \langle \tilde\omega\rangle^n_t:=\sum_{[u,v]\in\tilde\pi_n} (\tilde \omega(u\wedge t)-\tilde\omega(v\wedge t))^2,
  \end{align*}
  exists for every $t\in [0,T)$ and $\langle\tilde \omega \rangle_t\to \infty$ as $t\to T$. 
\end{lemma}

\begin{proof}
  Let $(\pi_n)$ be the refining sequence of partition given by the dyadic points $\mathbb{D}_n := \{k 2^{-n} \,:\,k \in \mathbb{N}_0\}$ with $\mathbb{N}_0:=\mathbb{N}\cup \{0\}$ adding the stopping time $\tau (\omega) := \inf \{ t>0 \,:\, \omega(t)=0 \}$ for $\omega\in C([0,T],\mathbb{R})$.
  
  Recalling the properties of a Brownian motion~$W$, we know that the event of a Brownian motion~$W$ starting at~$0$, $\tau\geq 1$, and $W_t\in (0,1)$ for $t\in (0,\tau)$ has a strictly positive probability. This fact ensures the existence of a constant $T_0>1$ and a (nowhere constant) function $f\in C^{\beta}([0,T_0],\mathbb{R})$ for every $\beta \in (0,1/2)$ such that 
  \begin{enumerate}
    \item $f(0)=f(T_0)=0$,
    \item $0\leq f(t)\leq 1$ for all $t\in [0,T_0]$,
    \item the pathwise quadratic variation given by $\langle f\rangle_t:=\lim_n\langle f \rangle_t^n$ exists along $(\pi_n)$ for every $t\in [0,T_0]$ (as limit in uniform convergence) and $\langle f \rangle_{T_0}>0$.
  \end{enumerate}
  Without loss of generality we may assume $T_0=1$ since it is always possible to modify~$f$ to ensure this without loosing the others properties. 
  
  Setting $T:= \sum_{n\in \mathbb{N}_0}n^{-2}<\infty$ and iteratively $t_n:=t_{n-1}+n^{-2}$ for $n\in \mathbb{N}$ with $t_0=0$, we define
  \begin{equation*}
    \tilde \omega (t):= n^{-1/2} f(n^2(t-t_{n-1}))\quad \text{for }t\in [t_{n-1},t_n),
  \end{equation*}
  with $\tilde \omega (T):=0$.
  
  Let us first show that $\tilde \omega \in C^{1/4}([0,T],\mathbb{R})$. For $s,t\in [0,T)$ there exist $n,m\in \mathbb{N}$ such that $s\in [t_{n-1},t_n]$ and $t\in [t_{m-1}, t_m]$. Therefore, we get 
  \begin{align*}
    |\tilde \omega (t)-&\tilde \omega (s)|
    \leq \sum_{k=n}^m |\tilde \omega (t_k\wedge t)-\tilde \omega (t_{k-1}\vee s)|\\
    &\leq |\tilde \omega (t_m\wedge t)-\tilde \omega (t_{m-1}\vee s)| +  |\tilde \omega (t_n\wedge t)-\tilde \omega (t_{n-1}\vee s)|\\
    &\leq L_f m^{-1/2} |m^{2}( (t_m\wedge t)- (t_{m-1}\vee s))|^{1/4} +  L_f n^{-1/2}  |n^{2}((t_n\wedge t)-\omega (t_{n-1}\vee s))|^{1/4}\\
    &\leq 2 L_f  |t-s|^{1/4},
  \end{align*}
  where $L_f>0$ denotes the $1/4$-H\"older norm of~$f$. If $0\leq s<t=T$, choose $n,m\in \mathbb{N}$ such that $s\in [t_{n-1},t_n]$  and $m^{-1/2}\leq |t-s|^{1/4}$. This time, we get 
  \begin{align*}
    |\tilde \omega (t)-\tilde \omega (s)|
    &\leq |\tilde \omega (T)-\tilde \omega (t_{m-1}\vee s)| +  |\tilde \omega (t_n\wedge t)-\tilde \omega (t_{n-1}\vee s)|\\
    &\leq m^{-1/2} +  L_f n^{-1/2}  |n^{2}((t_n\wedge t)-\omega (t_{n-1}\vee s))|^{1/4}\\
    &\leq  (1+L_f)  |t-s|^{1/4}.
  \end{align*}
  Based on these two estimates, we see that $\tilde \omega \in C^{1/4}([0,T],\mathbb{R})$.
  
  To obtain the desired properties of the quadratic variation, we define the partition $\tilde \pi_m$ for $m\in \mathbb{N}$ as follows.
  For $n\leq m$, $\tilde\pi_m$ restricted to $[t_{n-1},t_{n}]$ consists of the point 
  \begin{equation*}
    \tau_{k}^m:=\inf\{ t\geq \tau_{k-1}^m \,:\, \tilde \omega (t)= f(k2^{-m})\} \quad \text{and}\quad \tau_0^m := t_{n-1},
  \end{equation*}
  for $n\geq m$, choose $\tilde\pi_m$ restricted to $[t_{n-1},t_{n}]$ to be empty, and $T$ is included in $\tilde\pi_m$. Note that $(\tilde\pi_m)$ is a refining sequence of partitions. Furthermore, by contraction of $(\tilde\pi_m)$ the pathwise quadratic variation of $\tilde \omega$ exists along  $(\pi_m)$ for all $t\in [0,T)$ and for $t_n$ we observe that 
  \begin{equation*}
    \langle \tilde \omega \rangle_{t_n} = \sum_{k=1}^n \frac{\langle f\rangle_{T_0}}{n}
  \end{equation*}
  which goes to infinity as $t_n\to T$ or in other word $n\to \infty$.
\end{proof}


\begin{thebibliography}{49}
\providecommand{\natexlab}[1]{#1}
\providecommand{\url}[1]{\texttt{#1}}
\expandafter\ifx\csname urlstyle\endcsname\relax
  \providecommand{\doi}[1]{doi: #1}\else
  \providecommand{\doi}{doi: \begingroup \urlstyle{rm}\Url}\fi

\bibitem[Acciaio and Larsson(2017)]{Acciaio2017}
B.~Acciaio and M.~Larsson.
\newblock Semi-static completeness and robust pricing by informed investors.
\newblock \emph{Ann. Appl. Probab.}, 27\penalty0 (4):\penalty0 2270--2304,
  2017.

\bibitem[Acciaio et~al.(2016)Acciaio, Beiglb\"ock, Penkner, and
  Schachermayer]{Acciaio2016}
B.~Acciaio, M.~Beiglb\"ock, F.~Penkner, and W.~Schachermayer.
\newblock A model-free version of the fundamental theorem of asset pricing and
  the super-replication theorem.
\newblock \emph{Math. Finance}, 26\penalty0 (2):\penalty0 233--251, 2016.

\bibitem[Aksamit et~al.(2016)Aksamit, Hou, and Ob{\l}\'{o}j]{Aksamit2016}
A.~Aksamit, Z.~Hou, and J.~Ob{\l}\'{o}j.
\newblock Robust framework for quantifying the value of information in pricing
  and hedging.
\newblock \emph{Preprint arXiv:1605.02539}, 2016.

\bibitem[Avellaneda et~al.(1995)Avellaneda, Levy, and
  Par\'{a}s]{Avellaneda1995}
M.~Avellaneda, A.~Levy, and A.~Par\'{a}s.
\newblock Pricing and hedging derivative securities in markets with uncertain
  volatilities.
\newblock \emph{Appl. Math. Finance}, 2\penalty0 (2):\penalty0 73--88, 1995.

\bibitem[Bartl(2019)]{bartl2016exponential}
D.~Bartl.
\newblock Exponential utility maximization under model uncertainty for
  unbounded endowments.
\newblock \emph{Ann. Appl. Probab.}, 29\penalty0 (1):\penalty0 577--612, 2019.

\bibitem[Bartl et~al.(2017{\natexlab{a}})Bartl, Cheridito, Kupper, and
  Tangpi]{BCKT}
D.~Bartl, P.~Cheridito, M.~Kupper, and L.~Tangpi.
\newblock Duality for increasing convex functionals with countably many
  marginal constraints.
\newblock \emph{Banach J. Math. Anal.}, 11\penalty0 (1):\penalty0 72--89,
  2017{\natexlab{a}}.

\bibitem[Bartl et~al.(2017{\natexlab{b}})Bartl, Kupper, and Neufeld]{Bartl2017}
D.~Bartl, M.~Kupper, and A.~Neufeld.
\newblock {Pathwise superhedging on prediction sets}.
\newblock \emph{Preprint arXiv:1711.02764}, 2017{\natexlab{b}}.

\bibitem[Bartl et~al.(2019)Bartl, Cheridito, and Kupper]{Bartl2019}
D.~Bartl, P.~Cheridito, and M.~Kupper.
\newblock Robust expected utility maximization with medial limits.
\newblock \emph{Journal of Mathematical Analysis and Applications},
  471\penalty0 (1):\penalty0 752--775, 2019.

\bibitem[Bayer et~al.(2016)Bayer, Friz, and Gatheral]{Bayer2016}
C.~Bayer, P.~Friz, and J.~Gatheral.
\newblock Pricing under rough volatility.
\newblock \emph{Quant. Finance}, 16\penalty0 (6):\penalty0 887--904, 2016.

\bibitem[Beiglb\"ock and Siorpaes(2015)]{Beiglbock2015}
M.~Beiglb\"ock and P.~Siorpaes.
\newblock Pathwise versions of the {B}urkholder-{D}avis-{G}undy inequality.
\newblock \emph{Bernoulli}, 21\penalty0 (1):\penalty0 360--373, 2015.

\bibitem[Beiglb\"ock et~al.(2013)Beiglb\"ock, Henry-Labord\`ere, and
  Penkner]{bei-hl-pen}
M.~Beiglb\"ock, P.~Henry-Labord\`ere, and F.~Penkner.
\newblock Model-independent bounds for option prices -- a mass transport
  approach.
\newblock \emph{Finance Stoch.}, 17\penalty0 (3):\penalty0 477--501, 2013.

\bibitem[Beiglb{\"o}ck et~al.(2017)Beiglb{\"o}ck, Cox, Huesmann, Perkowski, and
  Pr{\"o}mel]{Beiglbock2017}
M.~Beiglb{\"o}ck, A.~M.~G. Cox, M.~Huesmann, N.~Perkowski, and D.~J.
  Pr{\"o}mel.
\newblock Pathwise superreplication via {V}ovk's outer measure.
\newblock \emph{Finance Stoch.}, 21\penalty0 (4):\penalty0 1141--1166, 2017.

\bibitem[Beiglb\"ock et~al.(2017)Beiglb\"ock, Nutz, and Touzi]{bei-nutz-tou}
M.~Beiglb\"ock, M.~Nutz, and N.~Touzi.
\newblock Complete duality for martingale transport on the line.
\newblock \emph{Ann. Probab.}, 45\penalty0 (6):\penalty0 3038--3074, 2017.

\bibitem[Bennedsen et~al.(2016)Bennedsen, Lunde, and Pakkanen]{Bennedsen2016}
M.~Bennedsen, A.~Lunde, and M.~S. Pakkanen.
\newblock {Decoupling the short- and long-term behavior of stochastic
  volatility}.
\newblock \emph{Preprint arXiv:1610.00332}, 2016.

\bibitem[Biagini et~al.(2017)Biagini, Bouchard, Kardaras, and
  Nutz]{Biagini2017}
S.~Biagini, B.~Bouchard, C.~Kardaras, and M.~Nutz.
\newblock Robust fundamental theorem for continuous processes.
\newblock \emph{Math. Finance}, 27\penalty0 (4):\penalty0 963--987, 2017.

\bibitem[Burzoni et~al.(2016)Burzoni, Frittelli, Hou, Maggis, and
  Ob{\l}{\'o}j]{burzoni2016pointwise}
M.~Burzoni, M.~Frittelli, Z.~Hou, M.~Maggis, and J.~Ob{\l}{\'o}j.
\newblock Pointwise arbitrage pricing theory in discrete time.
\newblock \emph{Preprint arXiv:1612.07618, forthcoming in Math. Oper. Res.},
  2016.

\bibitem[Burzoni et~al.(2017)Burzoni, Frittelli, and Maggis]{bur-fri-mag}
M.~Burzoni, M.~Frittelli, and M.~Maggis.
\newblock Model-free superhedging duality.
\newblock \emph{Ann. Appl. Probab.}, 27\penalty0 (3):\penalty0 1452--1477,
  2017.

\bibitem[Cheridito et~al.(2015)Cheridito, Kupper, and
  Tangpi]{cheridito2015representation}
P.~Cheridito, M.~Kupper, and L.~Tangpi.
\newblock Representation of increasing convex functionals with countably
  additive measures.
\newblock \emph{Preprint arXiv:1502.05763}, 2015.

\bibitem[Cheridito et~al.(2017)Cheridito, Kupper, and Tangpi]{RobHeding}
P.~Cheridito, M.~Kupper, and L.~Tangpi.
\newblock Duality formulas for robust pricing and hedging in discrete time.
\newblock \emph{SIAM J. Financial Math.}, 8\penalty0 (1):\penalty0 738--765,
  2017.

\bibitem[Delbaen and Schachermayer(2006)]{Delbaen2006}
F.~Delbaen and W.~Schachermayer.
\newblock \emph{The Mathematics of Arbitrage}.
\newblock Springer Finance. Springer-Verlag, Berlin, 2006.

\bibitem[Denis and Martini(2006)]{denis06}
L.~Denis and C.~Martini.
\newblock A theoretical framework for the pricing of contingent claims in the
  presence of model uncertainty.
\newblock \emph{Ann. Appl. Probab.}, 16\penalty0 (2):\penalty0 827--852, 2006.

\bibitem[Dolinsky and Soner(2014)]{Dolinsky2014}
Y.~Dolinsky and H.~M. Soner.
\newblock Martingale optimal transport and robust hedging in continuous time.
\newblock \emph{Probab. Theory Related Fields}, 160\penalty0 (1-2):\penalty0
  391--427, 2014.

\bibitem[Dolinsky and Soner(2015)]{DS_Skorokhod}
Y.~Dolinsky and H.~M. Soner.
\newblock Martingale optimal transport in the {S}korokhod space.
\newblock \emph{Stoch. Proc. Appl.}, 125\penalty0 (10):\penalty0 3893--3931,
  2015.

\bibitem[El~Euch and Rosenbaum(2019)]{ElEuch2019}
O.~El~Euch and M.~Rosenbaum.
\newblock The characteristic function of rough heston models.
\newblock \emph{Math. Finance}, 29\penalty0 (1):\penalty0 3--38, 2019.

\bibitem[Friz and Victoir(2010)]{Friz2010}
P.~Friz and N.~Victoir.
\newblock \emph{Multidimensional stochastic processes as rough paths. Theory
  and applications}.
\newblock Cambridge University Press, 2010.

\bibitem[Galichon et~al.(2014)Galichon, Henry-Labord\`ere, and
  Touzi]{gal-hl-tou}
A.~Galichon, P.~Henry-Labord\`ere, and N.~Touzi.
\newblock A stochastic control approach to no-arbitrage bounds given marginals,
  with an application to lookback options.
\newblock \emph{Ann. Appl. Probab.}, 24\penalty0 (1):\penalty0 312--336, 2014.

\bibitem[Gatheral et~al.(2018)Gatheral, Jaisson, and Rosenbaum]{Gatheral2014}
J.~Gatheral, T.~Jaisson, and M.~Rosenbaum.
\newblock Volatility is rough.
\newblock \emph{Quant. Finance}, 18\penalty0 (6):\penalty0 933--949, 2018.

\bibitem[Grisvard(1985)]{Grisvard1985}
P.~Grisvard.
\newblock \emph{{Elliptic Problems in Nonsmooth Domains}}, volume~24 of
  \emph{Monographs and Studies in Mathematics}.
\newblock Pitman (Advanced Publishing Program), Boston, MA, 1985.

\bibitem[Guo et~al.(2017)Guo, Tan, and Touzi]{Guo-Tan-Tou17}
G.~Guo, X.~Tan, and N.~Touzi.
\newblock Tightness and duality of martingale transport on the {S}korokhod
  space.
\newblock \emph{Stoch. Proc. Appl.}, 127\penalty0 (3):\penalty0 927--956, 2017.

\bibitem[Hobson(1998)]{Hob98}
D.~Hobson.
\newblock Robust hedging of the lookback option.
\newblock \emph{Finance Stoch.}, 2:\penalty0 329--347, 1998.

\bibitem[Hou and Ob{\l}\'{o}j(2018)]{Hou2015}
Z.~Hou and J.~Ob{\l}\'{o}j.
\newblock Robust pricing--hedging dualities in continuous time.
\newblock \emph{Finance Stoch.}, 22\penalty0 (3):\penalty0 511--567, 2018.

\bibitem[Jacod and Shiryaev(1998)]{jacod1998local}
J.~Jacod and A.~N. Shiryaev.
\newblock Local martingales and the fundamental asset pricing theorems in the
  discrete-time case.
\newblock \emph{Finance Stoch.}, 2\penalty0 (3):\penalty0 259--273, 1998.

\bibitem[Lyons(1995)]{Lyons1995}
T.~J. Lyons.
\newblock Uncertain volatility and the risk-free synthesis of derivatives.
\newblock \emph{Appl. Math. Finance}, 2\penalty0 (2):\penalty0 117--133, 1995.

\bibitem[Maligranda(1992)]{Maligranda1992}
L.~Maligranda.
\newblock Weakly compact operators and interpolation.
\newblock \emph{Acta Appl. Math.}, 27\penalty0 (1-2):\penalty0 79--89, 1992.
\newblock Positive operators and semigroups on Banach lattices (Cura\c{c}ao,
  1990).

\bibitem[Mao(2008)]{Mao2008}
X.~Mao.
\newblock \emph{Stochastic differential equations and applications}.
\newblock Horwood Publishing Limited, Chichester, second edition, 2008.

\bibitem[Mykland(2003)]{Mykland2003}
P.~A. Mykland.
\newblock Financial options and statistical prediction intervals.
\newblock \emph{Ann. Statist.}, 31\penalty0 (5):\penalty0 1413--1438, 2003.

\bibitem[Neufeld and Nutz(2013)]{Neufeld2013}
A.~Neufeld and M.~Nutz.
\newblock Superreplication under volatility uncertainty for measurable claims.
\newblock \emph{Electron. J. Probab.}, 18\penalty0 (48):\penalty0 14, 2013.

\bibitem[Nutz(2015)]{Nutz2015}
M.~Nutz.
\newblock Robust superhedging with jumps and diffusion.
\newblock \emph{Stoch. Proc. Appl.}, 125\penalty0 (12):\penalty0 4543--4555,
  2015.

\bibitem[Peng(2010)]{Peng2010}
S.~Peng.
\newblock Nonlinear expectation and stochastic calculus under uncertainty.
\newblock \emph{Preprint arXiv:1002.4546}, 2010.

\bibitem[Perkowski and Pr\"omel(2016)]{Perkowski2016}
N.~Perkowski and D.~J. Pr\"omel.
\newblock Pathwise stochastic integrals for model free finance.
\newblock \emph{Bernoulli}, 22\penalty0 (4):\penalty0 2486--2520, 2016.

\bibitem[Possama{\"{\i}} et~al.(2013)Possama{\"{\i}}, Royer, and
  Touzi]{Possamai2013}
D.~Possama{\"{\i}}, G.~Royer, and N.~Touzi.
\newblock On the robust superhedging of measurable claims.
\newblock \emph{Electron. Commun. Probab.}, 18\penalty0 (95):\penalty0 13,
  2013.

\bibitem[Rostek and Sch{\"o}bel(2013)]{Rostek2013}
S.~Rostek and R.~Sch{\"o}bel.
\newblock A note on the use of fractional {B}rownian motion for financial
  modeling.
\newblock \emph{Economic Modelling}, 30:\penalty0 30--35, 2013.

\bibitem[Schied and Voloshchenko(2016)]{Schied2016}
A.~Schied and I.~Voloshchenko.
\newblock Pathwise no-arbitrage in a class of {D}elta hedging strategies.
\newblock \emph{Probab. Uncertain. Quant. Risk}, 1:\penalty0 Paper No. 3, 25,
  2016.

\bibitem[Sion(1958)]{sion1958general}
M.~Sion.
\newblock On general minimax theorems.
\newblock \emph{Pacific J. Math}, 8\penalty0 (1):\penalty0 171--176, 1958.

\bibitem[Soner et~al.(2011)Soner, Touzi, and Zhang]{STZ10}
H.~M. Soner, N.~Touzi, and J.~Zhang.
\newblock Quasi-sure stochastic analysis through aggregation.
\newblock \emph{Electron. J. Probab.}, 16:\penalty0 no. 67, 1844--1879, 2011.

\bibitem[Soner et~al.(2012)Soner, Touzi, and Zhang]{STZ3}
H.~M. Soner, N.~Touzi, and J.~Zhang.
\newblock Wellposedness of second order backward {SDE}s.
\newblock \emph{Probab. Theory Related Fields}, 153\penalty0 (1-2):\penalty0
  149--190, 2012.

\bibitem[Soner et~al.(2013)Soner, Touzi, and Zhang]{STZ1}
H.~M. Soner, N.~Touzi, and J.~Zhang.
\newblock Dual formulation of second order target problems.
\newblock \emph{Ann. Appl. Probab.}, 23\penalty0 (2):\penalty0 308--347, 2013.

\bibitem[Vovk(2012)]{Vovk2012}
V.~Vovk.
\newblock Continuous-time trading and the emergence of probability.
\newblock \emph{Finance Stoch.}, 16\penalty0 (4):\penalty0 561--609, 2012.

\bibitem[Vovk(2016)]{Vovk2016}
V.~Vovk.
\newblock Another example of duality between game-theoretic and
  measure-theoretic probability.
\newblock \emph{Preprint arXiv:1608.02706}, 2016.

\end{thebibliography}

\vspace{0.7cm}
\noindent Daniel Bartl, Universit\"at Wien, Austria\\ 
{\small \textit{E-mail address:} daniel.bartl@univie.ac.at}\bigskip

\vspace{-.1cm}
\noindent Michael Kupper, Universit\"at Konstanz, Germany\\
{\small\textit{E-mail address:} kupper@uni-konstanz.de}\bigskip

\vspace{-.1cm}
\noindent David J. Pr\"omel, University of Oxford, United Kingdom\\
{\small\textit{E-mail address:} proemel@maths.ox.ac.uk}\bigskip

\vspace{-.1cm}
\noindent Ludovic Tangpi, Princeton University, United States of America\\ 
{\small\textit{E-mail address:} ludovic.tangpi@princeton.edu}

\end{document}